\pdfoutput=1

\documentclass[12pt]{article}

\usepackage{graphicx}

\usepackage{url} 

\usepackage{natbib}

\setcitestyle{authoryear,open={(},close={)}}

\newcommand{\blind}{0}

\usepackage{hyperref}
\usepackage{graphicx}
\usepackage{float}
\usepackage[section]{placeins}
\usepackage{setspace,dsfont,enumerate,cite,color,bm,float,epstopdf,verbatim,subfiles,lipsum}    
\usepackage{amsmath}
\usepackage{cleveref}
\usepackage{amssymb}
\usepackage[outercaption]{sidecap}
\usepackage{tabularx}
\usepackage[table]{xcolor}
\usepackage{pgffor}
\usepackage{comment}
\usepackage{algorithm}
\usepackage{algpseudocode}
\usepackage{caption}
\usepackage{subcaption}
\newcommand{\PP}{\mathbb{P}}

\usepackage{bbold}
\newcommand{\Nagr}{\mathcal{N}}
\newcommand{\gp}{\mathcal{GP}}
\newcommand{\EE}{\mathbb{E}}
\newcommand{\Lagr}{\mathcal{L}}
\newcommand{\Bagr}{\mathcal{B}}

\usepackage{amsthm}
\newtheorem{remark}{Remark}

\newtheorem{proposition}{Proposition}
\newtheorem{corollary}{Corollary}

\usepackage{scalerel,stackengine}
\stackMath
\newcommand\reallywidehat[1]{%
\savestack{\tmpbox}{\stretchto{%
  \scaleto{%
    \scalerel*[\widthof{\ensuremath{#1}}]{\kern-.6pt\bigwedge\kern-.6pt}%
    {\rule[-\textheight/2]{1ex}{\textheight}}
  }{\textheight}%
}{0.5ex}}%
\stackon[1pt]{#1}{\tmpbox}%
}
\usepackage{tikz-cd}

\usepackage{tikz}
\usetikzlibrary{arrows,chains,matrix,positioning,scopes}

\makeatletter
\tikzset{join/.code=\tikzset{after node path={%
\ifx\tikzchainprevious\pgfutil@empty\else(\tikzchainprevious)%
edge[every join]#1(\tikzchaincurrent)\fi}}}
\makeatother

\tikzset{>=stealth',every on chain/.append style={join},
         every join/.style={->}}
\tikzstyle{labeled}=[execute at begin node=$\scriptstyle,
   execute at end node=$]

\begin{document}

\def\spacingset#1{\renewcommand{\baselinestretch}
{#1}\small\normalsize} \spacingset{1}

\if0\blind
{
  \title{\bf Deep Gaussian Process Emulation and Uncertainty Quantification for Large Computer Experiments}
  \author{Faezeh Yazdi\thanks{Corresponding author, Department of Statistics and Actuarial Science, Simon Fraser University, fyazdi@sfu.ca}~~~
  Derek Bingham\thanks{Department of Statistics and Actuarial Science, Simon Fraser University}~~~
  Daniel Williamson\thanks{Department of Mathematics, University of Exeter}}
    
  \maketitle
} \fi

\bigskip
\begin{abstract}
Computer models are used as a way to explore complex physical systems. Stationary Gaussian process emulators, with their accompanying uncertainty quantification, are popular surrogates for computer models. However, many computer models are not well represented by stationary Gaussian processes models. Deep Gaussian processes have been shown to be capable of capturing non-stationary behaviors and abrupt regime changes in the computer model response. In this paper, we explore the properties of two deep Gaussian process formulations within the context of computer model emulation. For one of these formulations, we introduce a new parameter that controls the amount of smoothness in the deep Gaussian process layers. We adapt a stochastic variational approach to inference for this model, allowing for prior specification and posterior exploration of the smoothness of the response surface. Our approach can be applied to a large class of computer models, and scales to arbitrarily large simulation designs. The proposed methodology was motivated by the need to emulate an astrophysical model of the formation of binary black hole mergers. 
\end{abstract}

\vfill

\newpage
\spacingset{1.8} 
\section{Introduction}

Computer models, or simulators, are widely used as a way to explore complex physical systems. Frequently, a simulator is computationally expensive and only a limited number of model evaluations are available.  In other settings, the computational model is relatively fast to evaluate, but is not broadly available \citep[e.g.][]{Kauf}. In either case, an emulator of the computer models is needed to act as a surrogate.

Stationary Gaussian processes (GP) have become the conventional approach for deterministic computer model emulation \citep{Sacks,jones}. In many cases, the simulator response surface does not resemble a realization of a stationary GP. 

Recently, deep Gaussian processes (DGP) have been proposed for non-parametric regression \citep{law,Dunlop}, and 
computer model emulation \citep{ann,raja,rada,dw}. \citet{con} investigated the convergence properties DGP regression in general settings.

Large datasets are challenging for GPs and even more so for DGPs. The Vecchia approximation \citep{kat} has been successfully implemented for DGPs in current Bayesian implementations \citep[e.g.][]{min2,sau2} - scaling to hundreds of thousands of data points  but not necessarily millions as in our motivating application. In this work, we introduce a non-stationary DGP emulator which can be applied to a large class of complex computer models and scales to large simulation designs. 

Our first aim is to lay out some of the properties of DGPs for computer model emulation. We adapt the approach in \citet{Dunlop} so that it can serve as an emulator for computer models, and introduce a new parameter (or parameters) that allows us to control the smoothness of the DGP layers. We also develop a  variational inference (VI) approach for the proposed DGP which allows for prior specification and posterior exploration of the smoothness of the response surface. The proposed methodology was motivated by the emulation of an astrophysical model (Compact Object Mergers: Population Astrophysics and Statistics or COMPAS) that simulates the formation of binary black holes (BBHs) \citep{Bar}.

The paper is organized as follows: Section \ref{sec:app} introduces the challenges encountered when emulating the COMPAS model. Section \ref{sec:emu} presents an overview of GP emulation and describes existing work for non-stationary GPs. In Section \ref{sec:dgp}, the notation of two broad forms of the DGP are generalized to emphasize their differences and linking properties to stationary GPs. In Section \ref{sec:sdgp}, the DGP emulator is proposed through modifying one of the forms, followed by illustration of a new property induced by our modifications. We finish off this section with adapting a VI approach to our DGP emulator and providing a prediction method. In Sections \ref{sec:ill3} and \ref{num:com}, our proposed method is illustrated in a synthetic example as well as emulation of the COMPAS model. Section \ref{sec:sum3} concludes with a discussion. \\[-40pt]

\section{Application}
\label{sec:app} 

The application that motivated the proposed methodology was emulation of the COMPAS simulator of BBH mergers \citep{Bar}. BBHs are systems consisting of two black holes in close orbit around each other (Figure \ref{fig:fig0}). When two black holes merge, a gravitational wave signal is emitted and is measured by ground-based detectors \citep{Man}. During the merger, mass (i.e., the {\it{chirp mass}}) is expelled. The inputs to the COMPAS model describe the initial conditions of the star system (e.g., the mass of the most massive star) and parameters that govern physical processes. The output of interest is the chirp mass of the merger. Given the complexity of binary stellar evolution,  many billions of BBHs may need to be simulated to perform an experiment that is sufficiently large to make scientific inferences. Consequently, a fast statistical emulator of the COMPAS model is desired. 

\begin{figure}[h!]
\centering
\includegraphics[height=1.5in]{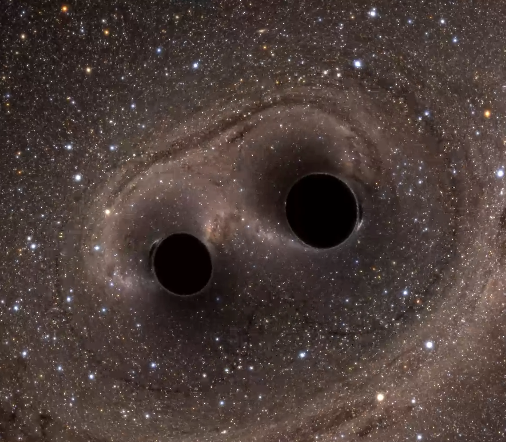}
\includegraphics[height=1.5in]{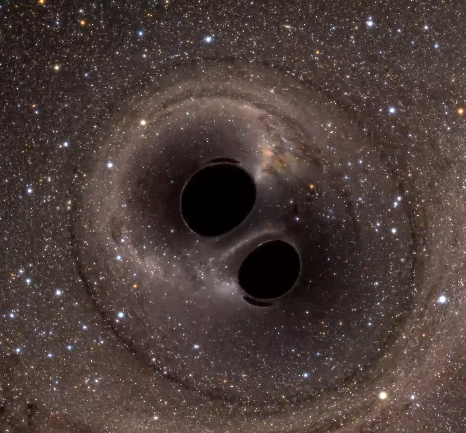}
\includegraphics[height=1.5in]{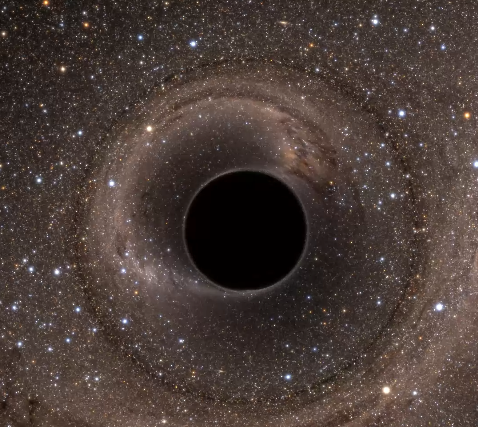}
\caption{Computer simulation of the black hole binary system GW150914. Credits: SXS (Simulating eXtreme Spacetimes) project}
\label{fig:fig0}
\end{figure}
The first challenge for emulation in this setting is the presence of regions of discontinuities in the response surface of the chirp mass. That is, BBHs only form in some regions of the input space. If a BBH system does not form, then no chirp mass is observed. The portions of input space that result in non-zero outputs are unknown. It is assumed that a chirp mass is observed in a union of compact regions within the input space. The second challenge for emulation is that the simulation suite is large ($\approx 2 \times 10^6$), and only about $24\%$ result in BBH mergers \citep[e.g.][]{broe}. If millions of BBHs are needed, a large number of simulations need to be performed. 

So, a fast emulator of the COMPAS model will have to be adapted to the type of complexity that we have in this setting and contend with a large number of simulation runs. This motivated us to propose a new methodology for building a non-stationary emulator which can be used in a wide range of applications - including the COMPAS model.\\[-40pt]

\section{Emulation}
\label{sec:emu}

The conventional approach to computer model emulation is to use a stationary GP \citep{Sacks}. Let $y=\eta(\mathbf{x})$ denote the scalar output of a deterministic computer model, $\eta(.)$, at input $\mathbf{x}\in \mathcal{X}\subseteq\mathbb{R}^d$. Let $\mathbf{X}$ be the $n_S \times d$ design matrix with outputs $\mathbf{y}=(y_1,\ldots, y_{n_S})^T$ (i.e., $y_i=\eta(\mathbf{x}_i ); \mbox{ } i=1,2,\ldots, n_S)$. The emulator is often specified as a stationary GP with constant mean, $\mu$, and   covariance \\[-25pt]
\begin{equation*}
\mathrm{Cov}(y(\mathbf{x}),y(\mathbf{x}^\prime))=k(\mathbf{x},\mathbf{x}^\prime; \bm{\phi})=\sigma^2\prod_{l=1}^{d}k_l(~|x_l-x_l^\prime|~;\lambda_l),
\label{eq:scov}
\end{equation*}\\[-35pt]
where $k:\mathcal{X}\times \mathcal{X}\rightarrow \mathbb{R}$, $k_l$ is a one-dimensional correlation function, and $\bm{\phi}=\{\sigma^2,\bm{\lambda}\}$ is the set containing  
the marginal variance, $\sigma^2$, and the correlation parameters, $\bm{\lambda}=(\lambda_1,\dots,\lambda_d)$. Conditioning on $\bm{\phi}$, $\mathbf{X}$, and $\mathbf{y}$, the predictive distribution at new input, $\mathbf{x}^*$,  is conditionally Gaussian with mean and variance, respectively, \\[-20pt]
\begin{equation}
\label{eqn:eq11}
k(\mathbf{X},\mathbf{x}^*;\bm{\phi})^T \mathbf{K}^{-1}\mathbf{y},~~~~~~k(\mathbf{X},\mathbf{x}^*;\bm{\phi})^T \mathbf{K}^{-1}k(\mathbf{X},\mathbf{x}^*;\bm{\phi}),
\end{equation}\\[-35pt]
where $\mathbf{K} = k(\mathbf{X},\mathbf{X};\bm{\phi})$ is the covariance matrix for the simulations, with $\mathbf{K}_{ij}=k(\mathbf{x}_i,\mathbf{x}_j;\bm{\phi})$, and $k(\mathbf{X},\mathbf{x}^*;\bm{\phi})$ is the $n_S \times 1$ vector of correlations between a response at $\mathbf{x}^*$ and those at the inputs in the design. 

There are situations where the simulator outputs are not well represented by a stationary GP. This can occur, when there are rapid changes in  the response surface or discontinuities. Figure \ref{fig:fig1} shows a model given by $\eta(x) = \mathbb{1}_{(0.3,0.7)}$ for $x\in(0,1)$ where a stationary GP has difficulty modeling the behaviour in the response \citep{Dunlop}. A more appropriate GP emulator of this model would locally adapt the correlation lengths to the response surface.

\begin{figure}[h!] 
  \centering
  
  \includegraphics[width=8cm]{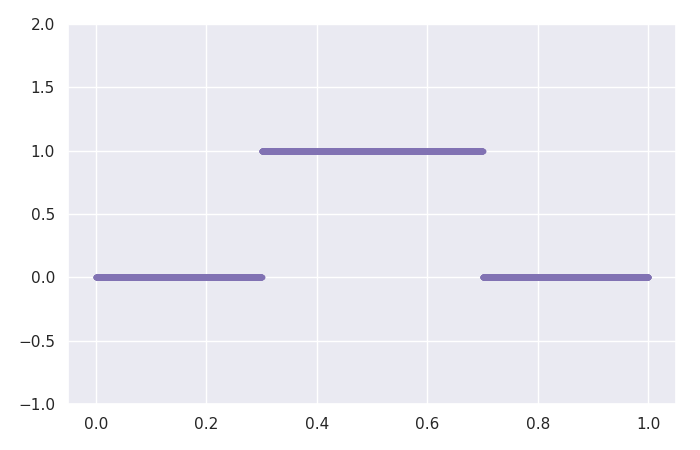}
  \caption{A simple, illustrative computer model with regions of discontinuities}
  \label{fig:fig1}
\end{figure}

Most approaches to non-stationary GP modeling can be viewed as either space-warping or covariance-modeling.  
In the former, the input space is warped so that the observations can be modeled as a stationary GP \citep[e.g.][]{Samp,Smit}. 
Alternatively, one can vary the correlation function over the input space to deal with non-stationarity \citep[e.g.][]{hig2,Paciorek}. In recent years, DGPs have been proposed to accommodate the non-stationary response surfaces, and in the next section we present two common formulations as generalizations of space-warping and covariance-modeling methods. 

\section{Deep Gaussian Processes (DGP)}
\label{sec:dgp}

\subsection{DGP Formulations}

The first DGP approach we consider \citep{law} uses compositions of GPs. Specifically, a DGP with $N$ hidden layers in this form is defined by composition of functions $f_1: \mathcal{X}\subseteq\mathbb{R}^d\rightarrow \mathbb{R}^{d^{\prime}_1}$ and $f_n: \mathbb{R}^{d^{\prime}_{n-1}}\rightarrow \mathbb{R}^{d^{\prime}_n}$ $(n=2,\dots,N+1)$, that are conditionally Gaussian,\\[-20pt]

\begin{equation*}
y(\mathbf{x})\vert f_{N}(\mathbf{x})~{\sim}~\gp(\mathbf{0}, k_{N+1}(f_{N}(\mathbf{x});\bm{\phi}_{N+1}))
\end{equation*}
\begin{equation}
f_{N,l}(\mathbf{x})\vert f_{N-1}(\mathbf{x})~{\sim}~\gp(\mathbf{0}, k_{N,l}(f_{N-1}(\mathbf{x});\bm{\phi}_{N})),~~~~~l=1,\dots,d^{\prime}_{N} ,
\label{eq:44}
\end{equation}\\[-65pt]
$$\vdots~~~~~~~~~$$\\[-70pt]
\begin{equation*}
f_{1,l}(\mathbf{x})~{\sim}~\gp(\mathbf{0}, k_{1,l}(\mathbf{x};\bm{\phi}_1)),~~~~l=1,\dots,d^{\prime}_{1},
\end{equation*}
 
where $d^{\prime}_{N+1}=1$, $\mathbf{x}\in \mathcal{X}\subseteq\mathbb{R}^d$, and $f_{n,l}(\mathbf{x})$ represents the $l^{\mathrm{th}}$ component of $f_n(\mathbf{x})\in \mathbb{R}^{d^{\prime}_n}$. Here $k_{1,l}:\mathcal{X}\times \mathcal{X}\rightarrow \mathbb{R}$ and $k_{n,l}:\mathbb{R}^{d^{\prime}_{n-1}}\times \mathbb{R}^{d^{\prime}_{n-1}}\rightarrow \mathbb{R}$ are stationary covariance functions, and $\bm{\phi}_{1}$ and $\bm{\phi}_{n}$ are covariance parameters for the GPs. Typically the covariance functions are chosen to be the same (e.g. squared exponential). Explicitly, the covariance function for subsequent layers ($n>1$) is $k_{n,l}(f_{n-1}(\mathbf{x}),f_{n-1}(\mathbf{x}');\bm{\phi}_n)$. Since the inputs to the covariance function at layer $n$ are a function of the outputs from the previous layer, this approach amounts to warping the input space. Here, $f_{N+1}$ refers to a DGP with $N$ hidden layers viewed as unobserved latent variables that warp the input space. 

An alternative form of the DGP was introduced in \citet{Dunlop}. In their specification, the covariance parameters of each layer are a function of the output of a previous layer. Specifically, a DGP with $N$ hidden layers in this form is defined by sequences of functions $f_n: \mathcal{X}\subseteq\mathbb{R}^d\rightarrow \mathbb{R}$ $(n=1,\dots,N+1)$, that are conditionally Gaussian, \\[-20pt]
\begin{equation*}
y(\mathbf{x})\vert f_{N}(\mathbf{x})~\sim~ \gp(\mathbf{0}, k_{N+1}(\mathbf{x};\bm{\phi}_{N+1}(f_{N}(\mathbf{x})))),
\label{eq:lik}
\end{equation*}
\begin{equation}
f_{N}(\mathbf{x})\vert f_{N-1}(\mathbf{x})~\sim~ \gp(\mathbf{0}, k_{N}(\mathbf{x};\bm{\phi}_{N}(f_{N-1}(\mathbf{x})))),
\label{eq:66}
\end{equation}\\[-65pt]
$$\vdots$$\\[-70pt]
\begin{equation*}
f_1(\mathbf{x})~\sim~\gp(\mathbf{0}, k_1(\mathbf{x};\bm{\phi}_1)),
\end{equation*}

where $\mathbf{x}\in \mathcal{X}\subseteq\mathbb{R}^d$. Similar to the previous approach, a stationary covariance function $k_1:\mathcal{X}\times \mathcal{X}\rightarrow \mathbb{R}$, with parameter $\bm{\phi}_1$, is used in the base layer. In this form, the covariance function $k_n:\mathcal{X}\times\mathcal{X}\rightarrow \mathbb{R}$ $(n=2,\dots,N+1)$ is a non-stationary covariance function which operates on the original input space $\mathcal{X}$, and $\bm{\phi}_{n}(f_{n-1}(\mathbf{x}))$ denotes all covariance parameters that are a function of the previous layer $f_{n-1}$. Explicitly, the covariance function for subsequent layers ($n>1$) is $k_n(\mathbf{x},\mathbf{x}';\bm{\phi}_n(f_{n-1}(\mathbf{x}),f_{n-1}(\mathbf{x}')))$ where $k_n(.,.)$ is always a function of the original input locations $\mathbf{x}$ and $\mathbf{x}'$, but the covariance parameters depend on the processes $f_{n-1}(\mathbf{x})$ and $f_{n-1}(\mathbf{x}')$ at those locations. Hence, handling non-stationarity in this form of the DGP can be considered as covariance-modelling. Under this specification, the observation vector
$\mathbf{y}=\mathbf{f}_{N+1}=f_{N+1}(\mathbf{X})\in \mathbb{R}^{n_S}$. Here $f_{N+1}$ refers to a DGP with $N$ hidden layers $(\mathbf{f}_1,\mathbf{f}_2,\dots, \mathbf{f}_{N})$, considered as unobserved random vectors which are used to discover the covariance among the given simulation outputs. For simplicity, we use $k_n(\mathbf{x},\mathbf{x}';f_{n-1}(\mathbf{x}),f_{n-1}(\mathbf{x}'))$ through the rest of the paper. 

\begin{remark}
Looking at (\ref{eq:44}) and (\ref{eq:66}), the two DGP formulations appear similar, but there are important differences. In equation (\ref{eq:66}) the input of each layer is always $\mathbf{x}$ in a $d$-dimensional space and is mapped to a $1$-dimensional output. The outputs of each hidden layer are considered as parameters and are used to model covariance parameters in the next layer. In equation (\ref{eq:44}) the input of layer $f_n$ ($n>1$) is the output of the previous layer $f_{n-1}$ in  $d'_{n-1}$-dimensional space and is mapped to a $d'_n$-dimensional output where $d'_n\geq d$ and $d'_n$ is not necessarily one (except $d'_{N+1}=1$). In this form, the outputs in each layer are considered as an warped domain space for the covariance function of the next layer.   
\end{remark}

\subsection{Stationarity as a special case}

One might ask what the DGP would look like if the data were actually a realization of a stationary GP. Proofs of the following two propositions and their corollary are given in the supplementary materials. 

\begin{proposition}

Under the DGP model in (\ref{eq:44}), $f_{n,l}(\mathbf{x})\vert f_{n-1}(\mathbf{x})$ $(l=1,\dots,d^{\prime}_{n})$ is a stationary GP if and only if  $f_{n-1}(\mathbf{x})=\mathbf{c}\odot\mathbf{x}$ for any $\mathbf{x}\in\mathcal{X}\subseteq\mathbb{R}^d$, where $\mathbf{c}\in\mathbb{R}^d$ is a vector of constants and $\odot$ represents an element-wise product. 
\label{eq:lfix}
\end{proposition}

\begin{proposition}
Under the DGP model in (\ref{eq:66}), $f_n(\mathbf{x})\vert f_{n-1}(\mathbf{x})$ is a stationary GP if and only if $f_{n-1}(\mathbf{x})$ is a constant for any $\mathbf{x}\in\mathcal{X}\subseteq\mathbb{R}^d$.
\label{eq:dfix}
\end{proposition}

\begin{corollary} 
Under the DGP model in (\ref{eq:66}), $y(.)$ is a realization of a stationary GP if and only if $y(.)$ is a DGP with one hidden layer where $f_1(\mathbf{x})$ is a constant for any $\mathbf{x}\in\mathcal{X}\subseteq\mathbb{R}^d$.\\[-20pt]
\label{eq:cor}
\end{corollary}
\begin{remark}\label{num:rem1}
Corollary \ref{eq:cor} has practical implications that are helpful when one wants to do preliminary exploration on a dataset. Specifically, to diagnose if a stationary GP is a good enough solution for the the given data, one can fit a one hidden layer DGP to the data.  If the first hidden layer is (almost) constant, a staionary GP will be an adequate solution (See chapter 3 of \citet{yazdi}).

\end{remark}

One advantage of the DGP formulation in (\ref{eq:66}) is that the covariance functions operate on the original input space rather than a warped input space. This allows us to more easily understand how the correlation changes through the original input space. Specifically, this exploration can be done by visualizing the layers versus the original inputs. We now describe a useful innovation
to this formulation and propose a VI approach for this model.\\[-40pt]

\section{DGP Emulation}\label{sec:sdgp}

In this section, new DGP methodology is proposed for computer model emulation. The proposed approach has two innovations $(i)$ modifying the DGP formulation in \citet{Dunlop} by introducing a new parameter that controls the smoothness of the DGP layers; $(ii)$ adapting a stochastic VI approach to the proposed DGP (this includes deriving a new evidence lower bound and sampling procedure for the proposd emulator). These innovations work together allow the DGP to capture a wide range of response surfaces that are not well-modelled by a stationary GP, and scales to large sample sizes. \\[-40pt]

\subsection{Proposed Model}

\subsubsection{Non-stationary Covariance Functions}
 
Let $\rho(.)$ be a stationary correlation function. The Non-stationary covariance function is introduced by \citet{Paciorek} using $\rho(.)$ as follows: \\[-20pt]
\begin{equation}
k(\mathbf{x},\mathbf{x}';\bm{\phi})=\sigma^2\frac{|\Sigma(\mathbf{x})|^{1/4}|\Sigma(\mathbf{x}^\prime)|^{1/4}}{|(\Sigma(\mathbf{x})+\Sigma(\mathbf{x}^\prime))/2|^{1/2}}~\rho(\sqrt{Q(\mathbf{x},\mathbf{x}^\prime)}),
\label{eq:paci}
\end{equation} 
where $\bm{\phi}=\{\sigma^2\}$ is the set containing the variance parameter, $\Sigma:\mathbb{R}^d\rightarrow \mathbb{R}^{d\times d}$ is a $d\times d$ matrix function of the inputs, and  \\[-10pt]
\begin{equation*}
Q(\mathbf{x},\mathbf{x}^\prime)=(\mathbf{x}-\mathbf{x}^\prime)^T\Bigg(\frac{\Sigma(\mathbf{x})+\Sigma(\mathbf{x}^\prime)}{2}\Bigg)^{-1}(\mathbf{x}-\mathbf{x}^\prime), \qquad \mathbf{x},\mathbf{x}^\prime\in\mathcal{X}\subseteq\mathbb{R}^d.
\label{eq:quad}
\end{equation*}
The quadratic form, $Q(.,.)$, allows the correlation between observations at the same Euclidean distance to vary across the input space.

\citet{Dunlop} use (\ref{eq:paci}) for each DGP hidden layer and specify the matrix $\Sigma(\mathbf{x})$ as \\[-30pt]
\begin{equation}
\Sigma(\mathbf{x})=H({f_{n-1}}(\mathbf{x}))\mathbf{I}_d,\\[-10pt]
\label{eq:sig}
\end{equation}
where $H:\mathbb{R}\rightarrow \mathbb{R}\geq0$ is a non-negative function, called the length scale function, and $\mathbf{I}_d$ is a $d \times d$ identity matrix. Note that $\Sigma(\mathbf{x})$ is a different matrix in each layer. Using (\ref{eq:paci}) and (\ref{eq:sig}), the non-stationary covariance functions $k_n(.,.)$ in equation (\ref{eq:66}) for $n>1$ can be written as \\[-15pt] 
\begin{equation}
k_n(\mathbf{x},\mathbf{x}';f_{n-1}(\mathbf{x}), f_{n-1}(\mathbf{x}'))=\sigma_n^2\frac{2^{d/2}[H(f_{n-1}(\mathbf{x}))]^{d/4}[H(f_{n-1}(\mathbf{x}'))]^{d/4}}{\big[H({f_{n-1}}(\mathbf{x}))+H({f_{n-1}}(\mathbf{x}'))\big]^{d/2}}\rho(\sqrt{Q(\mathbf{x},\mathbf{x}^\prime)}),\\[-10pt]
\label{eq:paci2}
\end{equation}
where 

\begin{equation}
\sqrt{Q(\mathbf{x},\mathbf{x}^\prime)}=\frac{\parallel \mathbf{x}-\mathbf{x}'\parallel_2}{\sqrt{H({f_{n-1}}(\mathbf{x}))+H({f_{n-1}}(\mathbf{x}'))/2}}.
\label{eq:rootsq}
\end{equation}
Note that, through the length scale function, $H(.)$, each layer of the DGP is used to evaluate the length-scale values at input locations for the subsequent layer. In this work, we select $\rho(.)$ to be a Matern correlation function with smoothness parameter $\nu>0$.

\subsubsection{Controlling the smoothness of the DGP layers }\label{sec:csdgp}

Though not discussed in \citet{Dunlop}, the choice of the length scale function, $H(.)$, in (\ref{eq:sig}) impacts the smoothness of the DGP layers. In this work, we propose \\[-30pt]
\begin{equation*}
H(f)=\mathrm{exp}(\alpha f), \\[-10pt]
\end{equation*}
where $\alpha$ is a parameter that controls the level of smoothness in the DGP layers. 

\begin{proposition}\label{pro:alpha}

With the kernel specified in (\ref{eq:paci2}) for the DGP in (\ref{eq:66}), let the length scale function, H(.), be $H(f_{n-1}(\mathbf{x}))=\mathrm{exp}(\alpha f_{n-1}(\mathbf{x}))$ for $n>1$ and $\alpha\geq0$. Then at any inputs $\mathbf{x},\mathbf{x}'\in\mathcal{X}\subseteq\mathbb{R}^d$ \\[-30pt]

$(a)$ $k_n(\mathbf{x},\mathbf{x}';f_{n-1}(\mathbf{x}),f_{n-1}(\mathbf{x}'))=\sigma_n^2~\rho(\parallel\mathbf{x}-\mathbf{x}^\prime\parallel_2)$, when $\alpha=0$,

$(b)$ $k_n(\mathbf{x},\mathbf{x}';f_{n-1}(\mathbf{x}),f_{n-1}(\mathbf{x}'))\rightarrow 0$, as $\alpha\rightarrow\infty$. \\[-35pt]

\end{proposition}

\begin{proof}
See the supplementary materials for details.
\end{proof}

The implication of $(a)$ is that $\alpha=0$ results in a stationary covariance function, and the implication of $(b)$ is that as $\alpha$ increases, the covariance between nearby observations decreases and the resulting functions will be less smooth. 

To explore the effect of $\alpha$ and $\lambda$  on $\rho(.)$, we simulate $500$ realisations of a 7 layer DGP for each combination of  $\alpha \in \{0.1,1,2,3\}$ and $\lambda \in \{0.1, 0.5, 1,2\}$. We then sum the absolute values of their 2nd derivatives (over the 200 design points) as a score for measuring the level of the smoothness and average these across the simulations. The results are given in Table \ref{table:sum} and indicate that $\lambda$ has a negligible impact on the smoothness relative to $\alpha$. We have found this to be true over many simulations and therefore develop inference for $\alpha$ for this formulation of the DGP. 

\begin{table}[h!]
		\caption{Average of sums of $\mid  d^{2}f_6/d{x^2}\mid$}
		\centering
\begin{tabular}{c c c c c}
			\hline\hline
			$\lambda$ & $\alpha=0.1$ &  $\alpha=1$ & $\alpha=2$ & $\alpha=3$ \\ [0.2ex]
			\hline
			$0.1$ & $2,327.4$ &  $3,183.7$ & $18,616.8$ & $226,127.6$\\
			$0.5$& $2,342.7$ &  $3,100.3$ & $16,246.2$ & $217,751.7$\\
			$1$& $2,323.2$ &  $2,975.7$ & $17,352.1$ & $209,352.8$ \\
			$2$& $2,332.6$ &  $3,167.3$& $16,189.5$ & $220,885.9$ \\[0.2ex]
			\hline \\[-10pt]
		\end{tabular}
\label{table:sum}
	\end{table} 

\subsection{Inference}\label{sec:inf}

\citet{har} showed that for moderate to large dimensional input spaces ($d>10$), the required number of runs for adequate prediction accuracy grows exponentially for stationary GPs. This problem can only be compounded for estimating DGPs when the response is non-stationary. Thus, one should expect very lage sample sizes for most DGP settings with moderate to large input dimensions. We present an extension to the doubly stochastic variational inference method (DSVI) \citet{Sal} to allow inference the DGP model with large datasets such as the one outlined in Section \ref{sec:app}. 

\subsubsection{Extension to DSVI}

For the proposed DGP, the key inference innovation is to derive an evidence lower bound (ELBO). The ELBO as an objective function is used to estimate all parameters needed for approximating the target posterior distribution. Our key modifications through deriving the ELBO are (1) choosing a GP prior distribution with a zero mean function in each layer, (2) adding vector parameters which are used to evaluate the length scale values at inducing locations in each layer, and (3) including inference on the smoothness parameter, $\alpha$. 

Let $\mathbf{X}=\big[\mathbf{x}_1,\dots,\mathbf{x}_{n_S}\big]^T$ be the $n_S \times d$ design matrix and $\mathbf{y}=(y_1,\ldots, y_{n_S})^T$ be the corresponding outputs of the simulator. We construct our DGP emulator as described in the previous section, with $N$ hidden layers, $\mathbf{f}_1,\mathbf{f}_2,\dots, \mathbf{f}_{N}$, for the computer model $\eta(.)$, where $\mathbf{f}_{n}=f_{n}(\mathbf{X})\in \mathbb{R}^{n_S}$ $(n=1,\dots,N)$ are unobserved random vectors.

Since $n_S$ is assumed to be large, inference on the GPs in the various layers will be computationally challenging. To address this, we propose to use DSVI \citep{Sal}. To do so, in each layer we define an additional set of $m$ inducing locations where $m\ll n_S$ (i.e. $\mathbf{Z}_1,\mathbf{Z}_2,\dots,\mathbf{Z}_N$ such that $\mathbf{Z}_n\in\mathbb{R}^{m\times d} $ for $n=1,\dots,N)$. Define inducing variables $(\mathbf{u}_1,\mathbf{u}_2,\dots,\mathbf{u}_N)$, to be the value of the hidden layers at the corresponding inducing locations, ($\mathbf{Z}_1,\mathbf{Z}_2,\dots,\mathbf{Z}_N$). In the DSVI, there is no dependency between the $\mathbf{u}_n$ and $\mathbf{u}_{n-1}$. This simplification allows for a fast inference \citep{Sal}. Following the DGP formulation in (\ref{eq:66}), in each layer, $f_n(\cdot)$ is a GP with zero prior mean.

To ease notation, hereafter $\PP(.)$ and $\Nagr(.,.)$ denote probability and normal distributions, respectively. The joint distribution of the first layer, and the corresponding inducing variables at the design points, and the inducing locations can be factorized as \\[-20pt]
\begin{equation}
\PP(\mathbf{f}_1,\mathbf{u}_1;\mathbf{X},\mathbf{Z}_1)=\PP(\mathbf{f}_1\vert\mathbf{u}_1;\mathbf{X},\mathbf{Z}_1)\PP(\mathbf{u}_1;\mathbf{Z}_1),
\label{eq:jp1}
\end{equation}
where $\PP(\mathbf{f}_1\vert\mathbf{u}_1;\mathbf{X},\mathbf{Z}_1)=\Nagr(\mathbf{f}_1\vert \boldsymbol{\mu}_1,\mathbf{\Sigma}_1)$ and $\PP(\mathbf{u}_1;\mathbf{Z}_1)=\Nagr(\mathbf{u}_1\vert \mathbf{0},k_1(\mathbf{Z}_1,\mathbf{Z}_1;\bm{\phi}_1))$, $k_1(.,.)$ is a stationary covariance function and, for $i,j=1,\dots,n_S$, \\[-20pt]
$$[{\boldsymbol{\mu}_1}]_i=\boldsymbol{\Gamma}_1(\mathbf{x}_i)^T\mathbf{u}_1,$$
$$[\mathbf{\Sigma}_1]_{ij}=k_1(\mathbf{x}_i,\mathbf{x}_j;\bm{\phi}_1)-\boldsymbol{\Gamma}_1(\mathbf{x}_i)^T k_1(\mathbf{Z}_1,\mathbf{Z}_1;\bm{\phi}_1)\boldsymbol{\Gamma}_1(\mathbf{x}_j),$$\\[-30pt]
and $\boldsymbol{\Gamma}_1(\mathbf{x}_i)=k_1(\mathbf{Z}_1,\mathbf{Z}_1;\bm{\phi}_1)^{-1}k_1(\mathbf{Z}_1,\mathbf{x}_i;\bm{\phi}_1)$. In subsequent layers we specify the non-stationary covariance functions, $k_n(.,.)$, from (\ref{eq:paci2}), for each $f_n(.)$. As before, the length scale function is $H(f)=\mathrm{exp}(\alpha f)$. In this setting, for $n=2,\dots,N$ the joint GP prior distribution is factorized as 
\\[-20pt] 
\begin{equation}
\PP(\mathbf{f}_n,\mathbf{u}_n;\mathbf{f}_{n-1},\mathbf{X},\mathbf{Z}_n, \boldsymbol{\delta}_{\mathbf{Z}_n},\alpha)=\PP(\mathbf{f}_n\vert\mathbf{u}_n;\mathbf{f}_{n-1},\mathbf{X},\mathbf{Z}_n, \boldsymbol{\delta}_{\mathbf{Z}_n},\alpha)\PP(\mathbf{u}_n\vert\mathbf{Z}_n,\boldsymbol{\delta}_{\mathbf{Z}_n},\alpha)\PP(\alpha),
\label{eq:jp2}
\end{equation}
where $\PP(\alpha)$ is the prior of the new parameter $\alpha$,  $\PP(\mathbf{f}_n\vert\mathbf{u}_n;\mathbf{f}_{n-1},\mathbf{X},\mathbf{Z}_n,\boldsymbol{\delta}_{\mathbf{Z}_n},\alpha)=\Nagr(\mathbf{f}_n\vert \boldsymbol{\mu}_n,\mathbf{\Sigma}_n)$ and $\PP(\mathbf{u}_n\vert\mathbf{Z}_n,\boldsymbol{\delta}_{\mathbf{Z}_n},\alpha)=\Nagr(\mathbf{u}_n\vert \mathbf{0},k_n(\mathbf{Z}_n,\mathbf{Z}_n;\boldsymbol{\delta}_{\mathbf{Z}_n},\alpha))$. 

Also, since there is no dependency between $\mathbf{u}_n$ and $\mathbf{u}_{n-1}$, we introduce $\boldsymbol{\delta}_{\mathbf{Z}_n}\in\mathbb{R}^m$ as a unknown vector parameter representing optimal value for $\mathbf{u}_{n-1}$ which is used to evaluate the length scale value at the inducing location $\mathbf{Z}_n$ through $H(.)$. For $i,j=1,\dots,n_S$ \\[-35pt]
$$[{\boldsymbol{\mu}_n}]_i=\boldsymbol{\Gamma}_n(\mathbf{x}_i,\mathbf{f}_{n-1}^i)^T\mathbf{u}_n,$$
$$[\mathbf{\Sigma}_n]_{ij}=k_n(\mathbf{x}_i,\mathbf{x}_j;\mathbf{f}_{n-1}^i,\mathbf{f}_{n-1}^j,\alpha)-\boldsymbol{\Gamma}_n(\mathbf{x}_i,\mathbf{f}_{n-1}^i)^T k_n(\mathbf{Z}_n,\mathbf{Z}_n;\boldsymbol{\delta}_{\mathbf{Z}_n},\alpha)\boldsymbol{\Gamma}_n(\mathbf{x}_j,\mathbf{f}_{n-1}^j),$$\\[-30pt]
and $\boldsymbol{\Gamma}_n(\mathbf{x}_i,\mathbf{f}_{n-1}^i)=k_n(\mathbf{Z}_n,\mathbf{Z}_n;\boldsymbol{\delta}_{\mathbf{Z}_n},\alpha)^{-1}k_n(\mathbf{Z}_n,\mathbf{x}_i;\boldsymbol{\delta}_{\mathbf{Z}_n},\mathbf{f}_{n-1}^i,\alpha)$ for 
$\mathbf{f}_{n}^i:=(\mathbf{f}_{n})_i=f_{n}(\mathbf{x}_i)$. Using (\ref{eq:jp1}) and (\ref{eq:jp2}) the joint density of the outputs and parameters is estimated as \\[-20pt]
\begin{equation}
\PP(\mathbf{y},\{\mathbf{f}_n, \mathbf{u}_n\}_{n=1}^{N},\alpha)=\PP(\mathbf{y}\vert \mathbf{f}_N)\PP(\mathbf{f}_1,\mathbf{u}_1;\mathbf{X},\mathbf{Z}_1)\prod_{n=2}^{N}\PP(\mathbf{f}_n,\mathbf{u}_n;\mathbf{f}_{n-1},\mathbf{X},\mathbf{Z}_n,\boldsymbol{\delta}_{\mathbf{Z}_n},\alpha)\\[-20pt]
\label{eq:jp3}
\end{equation}
$$=\PP(\mathbf{y}\vert \mathbf{f}_N)\PP(\mathbf{f}_1\vert\mathbf{u}_1;\mathbf{X},\mathbf{Z}_{1}) \PP(\mathbf{u}_1;\mathbf{Z}_{1})\Big(\prod_{n=2}^{N}\PP(\mathbf{f}_n\vert\mathbf{u}_n;\mathbf{f}_{n-1},\mathbf{X},\mathbf{Z}_{n},\boldsymbol{\delta}_{\mathbf{Z}_n},\alpha) \PP(\mathbf{u}_n\vert\mathbf{Z}_{n},\boldsymbol{\delta}_{\mathbf{Z}_n},\alpha)\Big)\PP(\alpha).$$
 
Similar to \citet{Sal}, we assume that the posterior distribution of $\{\mathbf{u}_n\}_{n=1}^{N}$ is factorized between the layers, along with maintaining the conditional dependency between $\mathbf{f}_n$ and $\mathbf{u}_n$ as in the model (\ref{eq:jp3}). We also include the variational posterior distribution of the new parameter $\alpha$, $q(\alpha)$, in our DGP variational posterior. Therefore, the DGP variational posterior is \\[-20pt]
\begin{equation}
q(\{\mathbf{f}_n, \mathbf{u}_n\}_{n=1}^{N},\alpha)=\PP(\mathbf{f}_1\vert\mathbf{u}_1;\mathbf{X},\mathbf{Z}_{1}) q(\mathbf{u}_1)\Big(\prod_{n=2}^{N}\PP(\mathbf{f}_n\vert\mathbf{u}_n;\mathbf{f}_{n-1},\mathbf{X},\mathbf{Z}_{n},\boldsymbol{\delta}_{\mathbf{Z}_n},\alpha) q(\mathbf{u}_n)\Big) q(\alpha)
\label{eq:12al}
\end{equation}
where $q(\mathbf{u}_n)$ is chosen to be $\Nagr(\mathbf{m}_n,\mathbf{s}_n)$ such that $\mathbf{m}_n\in\mathbb{R}^{m}$ and $\mathbf{s}_n\in\mathbb{R}^{m\times m}$ $(n=1,\dots,N)$ and $q(\alpha)=\Nagr(m_{\alpha},s_{\alpha})$ for  parameters $m_{\alpha}$ and $s_{\alpha}$. With this specification of $q(\mathbf{u}_n)$, the inducing variables can be marginalized from each layer analytically as \\[-35pt]
\begin{align}
\label{eqn:mar1}
\begin{split}
q(\mathbf{f}_1\vert\mathbf{m}_1,\mathbf{s}_1,\mathbf{X},\mathbf{Z}_{1})&=\int \PP(\mathbf{f}_1\vert\mathbf{u}_1;\mathbf{X},\mathbf{Z}_{1}) q(\mathbf{u}_1)~d\mathbf{u_1}=\Nagr(\mathbf{f}_1 \vert \tilde{\boldsymbol{\mu}}_1,\tilde{\mathbf{\Sigma}}_1),\\
q(\mathbf{f}_2\vert\mathbf{m}_2,\mathbf{s}_2,\mathbf{f}_{1},\mathbf{X},\mathbf{Z}_{2},\boldsymbol{\delta}_{\mathbf{Z}_2},\alpha)&=\int \PP(\mathbf{f}_2\vert\mathbf{u}_2;\mathbf{f}_{1},\mathbf{X},\mathbf{Z}_{2},\boldsymbol{\delta}_{\mathbf{Z}_2},\alpha)q(\mathbf{u}_2)~d\mathbf{u}_2\\
&=\Nagr(\mathbf{f}_2\vert \tilde{\boldsymbol{\mu}}_2,\tilde{\mathbf{\Sigma}}_2),
\\[-20pt]
\end{split}
\end{align}
$$\vdots~~~~~~~~~$$\\[-70pt]
\begin{align*}
\begin{split}
q(\mathbf{f}_N\vert\mathbf{m}_N,\mathbf{s}_N,\mathbf{f}_{N-1},\mathbf{X},\mathbf{Z}_{N},\boldsymbol{\delta}_{\mathbf{Z}_N},\alpha)&=\int \PP(\mathbf{f}_N\vert\mathbf{u}_N;\mathbf{f}_{N-1},\mathbf{X},\mathbf{Z}_{N},\boldsymbol{\delta}_{\mathbf{Z}_N},\alpha)q(\mathbf{u}_N) q(\alpha)~d\mathbf{u}_N\\
&=\Nagr(\mathbf{f}_N\vert \tilde{\boldsymbol{\mu}}_N,\tilde{\mathbf{\Sigma}}_N)q(\alpha),
\end{split}
\end{align*}
where for $i,j=1,\dots,n_S$ \\[-25pt]
$$[\tilde{\boldsymbol{\mu}}_1]_i={\mu}_{\mathbf{m}_1,\mathbf{Z}_1}(\mathbf{x}_i)=\boldsymbol{\Gamma}_1(\mathbf{x}_i)^T\mathbf{m}_1,
$$
\begin{equation}
\label{eqn:mean}
[\tilde{\boldsymbol{\mu}}_n]_i={\mu}_{\mathbf{m}_n,\mathbf{Z}_n,\boldsymbol{\delta}_{\mathbf{Z}_n},\alpha}(\mathbf{x}_i,\mathbf{f}_{n-1}^i)=\boldsymbol{\Gamma}_n(\mathbf{x}_i,\mathbf{f}_{n-1}^i)^T\mathbf{m}_n,
\end{equation}
$$[\tilde{\mathbf{\Sigma}}_1]_{ij}={\Sigma}_{\mathbf{s}_1,\mathbf{Z}_1}(\mathbf{x}_i,\mathbf{x}_j)=k_1(\mathbf{x}_i,\mathbf{x}_j;\bm{\phi}_1)-\boldsymbol{\Gamma}_1(\mathbf{x}_i)^T\Big[k_1(\mathbf{Z}_1,\mathbf{Z}_1;\bm{\phi}_1)-\mathbf{s}_1\Big]\boldsymbol{\Gamma}_1(\mathbf{x}_j),$$ 
$$[\tilde{\mathbf{\Sigma}}_n]_{ij}=k_n(\mathbf{x}_i,\mathbf{x}_j;\mathbf{f}_{n-1}^i,\mathbf{f}_{n-1}^j,\alpha)-\boldsymbol{\Gamma}_n(\mathbf{x}_i,\mathbf{f}_{n-1}^i)^T\Big[k_n(\mathbf{Z}_n,\mathbf{Z}_n;\boldsymbol{\delta}_{\mathbf{Z}_n},\alpha)-\mathbf{s}_n\Big]\boldsymbol{\Gamma}_n(\mathbf{x}_j,\mathbf{f}_{n-1}^j),$$ 
which is $[\tilde{\mathbf{\Sigma}}_n]_{ij}={\Sigma}_{\mathbf{s}_n,\mathbf{Z}_n,\boldsymbol{\delta}_{\mathbf{Z}_n},\alpha}(\mathbf{x}_i,\mathbf{x}_j;\mathbf{f}_{n-1}^i,\mathbf{f}_{n-1}^j)$ --- see the supplementary materials for a derivation of (\ref{eqn:mar1}). Therefore, using (\ref{eqn:mar1}) we obtain\\[-20pt]
\begin{equation}
q(\{\mathbf{f}_n\}_{n=1}^{N},\alpha)=\Big(\prod_{n=1}^{N}\Nagr(\mathbf{f}_n \vert \tilde{\boldsymbol{\mu}}_n,\tilde{\mathbf{\Sigma}}_n)\Big)q(\alpha).
\label{eq:13al}
\end{equation}

The ELBO of our DGP, $\Lagr_{DGP}$, is obtained as following\\[-20pt]
$$\Lagr_{DGP}=\EE_{q(\{\mathbf{f}_n, \mathbf{u}_n\}_{n=1}^{N},\alpha)}\textrm{log}\Bigg(\frac{\PP(\mathbf{y},\{\mathbf{f}_n, \mathbf{u}_n\}_{n=1}^{N},\alpha)}{q(\{\mathbf{f}_n, \mathbf{u}_n\}_{n=1}^{N},\alpha)}\Bigg)$$\\[-30pt]
and by substituting (\ref{eq:jp3}), (\ref{eq:12al}), it can be formed as \\[-20pt]
$$\Lagr_{DGP}=\sum_{i=1}^{n_S} \EE_{q(\{\mathbf{f}_n\}_{n=1}^{N},\alpha)}\Big(~\textrm{log}~\PP(y_i\vert \mathbf{f}_{N}^i)~\Big) - \textrm{KL}\Big(q(\mathbf{u}_1)\parallel \PP(\mathbf{u}_1;\mathbf{Z}_{1})\Big)\\[-11pt]$$
\begin{equation}
- \EE_{q(\alpha)}\Big[ \sum_{n=2}^{N}\textrm{KL}\Big(q(\mathbf{u}_n) \parallel \PP(\mathbf{u}_n;\mathbf{Z}_{n},\boldsymbol{\delta}_{\mathbf{Z}_n},\alpha)\Big)\Big]-\textrm{KL}\Big(q(\alpha) \parallel \PP(\alpha)\Big).
\label{eq:dgpelbo}
\end{equation}
See the supplementary materials for a derivation of the ELBO. \\[-40pt]
\subsubsection{Evaluating the Evidence Lower Bound}

To evaluate the ELBO, we must compute the first expectation term at each design point $\mathbf{x}_i$ for $i=1,\dots,n_S$. That is,\\[-20pt]
$$E_i=\EE_{q(\{\mathbf{f}_n\}_{n=1}^{N},\alpha)}\Big(~\textrm{log}~\PP(y_i\vert \mathbf{f}_{N}^i)~\Big)=\EE_{q(\{\mathbf{f}_n\}_{n=1}^{N},\alpha)}\Big(~\textrm{log}~\PP(y_i\vert f_{N}(\mathbf{x}_i))~\Big).$$
$E_i$ can be approximated by Monte Carlo (MC) using a sample from the variational posterior in (\ref{eq:13al}). We do this using univariate Gaussian distributions through the re-parameterization trick \citep{kin}. Specifically, we first sample $\alpha_t\sim\Nagr(m_{\alpha},s_{\alpha})=q(\alpha)$ and $(\epsilon_{n}^i)_t\sim \Nagr(0,1)$, where $t=1,\dots,T$ and $n=1,\dots,N$ represent indices of MC sampling iterations and number of DGP layers, respectively. Then, the sampled variables $(\hat{\mathbf{{f}}}_{1}^i)_t\sim q(\mathbf{f}_{1}^i\vert\mathbf{m}_1,\mathbf{s}_1,\mathbf{x}_i,\mathbf{Z}_{1}) $ and $(\hat{\mathbf{{f}}}_{n}^i)_t\sim q(\mathbf{f}_{n}^i\vert\mathbf{m}_n,\mathbf{s}_n,(\hat{\mathbf{{f}}}_{n-1}^i)_t,\mathbf{x}_i,\mathbf{Z}_{n},\boldsymbol{\delta}_{\mathbf{Z}_n},\alpha_t)$ $(n=2,\dots,N)$ are drawn
recursively as \\[-30pt]
\begin{equation*}
(\hat{\mathbf{{f}}}_{1}^i)_t={\mu}_{\mathbf{m}_1,\mathbf{Z}_1}(\mathbf{x}_i)+ (\epsilon_{1}^i)_t \sqrt{{\Sigma}_{\mathbf{s}_1,\mathbf{Z}_1}(\mathbf{x}_i,\mathbf{x}_i)},
\label{eq:15al}
\end{equation*}
\begin{equation*}
(\hat{\mathbf{{f}}}_{n}^i)_t={\mu}_{\mathbf{m}_n,\mathbf{Z}_n, \boldsymbol{\delta}_{\mathbf{Z}_n},\alpha_t}(\mathbf{x}_i,(\hat{\mathbf{{f}}}_{n-1}^i)_t)+ (\epsilon_{n}^i)_t \sqrt{{\Sigma}_{\mathbf{s}_n,\mathbf{Z}_n,\boldsymbol{\delta}_{\mathbf{Z}_n},\alpha_t}(\mathbf{x}_i,\mathbf{x}_i;(\hat{\mathbf{{f}}}_{n-1}^i)_t)},
\label{eq:15all}
\end{equation*}
where $\hat{\mathbf{{f}}}_{n}^i=\hat{f}_n(\mathbf{x}^i)\in\mathbb{R}$. At each input $\mathbf{x}_i$, this procedure is repeated $T$ times to obtain an unbiased estimate by taking a Monte Carlo estimate, i.e. \\[-15pt]
$$E_i\approx\frac{1}{T}\sum_{t=1}^T{\textrm{log}~\PP\Big(y_i~\vert ~(\hat{\mathbf{{f}}}_{N}^i)_t\Big)},\\[-5pt]$$
where $(\hat{\mathbf{{f}}}_{N}^i)_t={\mu}_{\mathbf{m}_N,\mathbf{Z}_N,\alpha_t}(\mathbf{x}_i,(\hat{\mathbf{{f}}}_{N-1}^i)_t)+ (\epsilon_{N}^i)_t \sqrt{{\Sigma}_{\mathbf{s}_{N},\mathbf{Z}_{N},\boldsymbol{\delta}_{\mathbf{Z}_n},\alpha_t}(\mathbf{x}_i,\mathbf{x}_i;(\hat{\mathbf{{f}}}_{N-1}^i)_t)}$. Also all the $\textrm{KL}$ terms in the ELBO can be computed analytically, and \\[-20pt]
$$ \EE_{q(\alpha)}\Big[ \sum_{n=2}^{N}\textrm{KL}\Big(q(\mathbf{u}_n) \parallel \PP(\mathbf{u}_n;\mathbf{Z}_{n},\boldsymbol{\delta}_{\mathbf{Z}_n},\alpha)\Big)\Big]$$ is estimated by sampling $\alpha\sim\Nagr(m_{\alpha},s_{\alpha})$. 

To achieve scalability for large $n_S$, the sum over $E_i$ can be estimated using sub-sampling. That is, \\[-40pt]
\begin{align}
\begin{split}
\Lagr_{DGP}&\approx\frac{n_S}{|\Bagr|}\sum_{i\in |\Bagr|}E_i - \textrm{KL}\Big(q(\mathbf{u}_1)\parallel \PP(\mathbf{u}_1;\mathbf{Z}_{1})\Big) - \EE_{q(\alpha)}\Big[ \sum_{n=2}^{N}\textrm{KL}\Big(q(\mathbf{u}_n) \parallel \PP(\mathbf{u}_n;\mathbf{Z}_{n},\boldsymbol{\delta}_{\mathbf{Z}_n},\alpha)\Big)\Big]\\
&-\textrm{KL}\Big(q(\alpha) \parallel \PP(\alpha)\Big),\\[-60pt]
\label{eqn:ev2}
\end{split}
\end{align} 
where $\Bagr$ represents a batch or sub-sample of data. The reason for using sub-sampling is that it decreases the time needed to evaluate the ELBO and thus perform optimization. Scalability is an imporant  aspect of DSVI, which is preserved in our modified version of the ELBO in (\ref{eqn:ev2}). 

With that said, to approximate the ELBO in our settings, three sources of stochasticity are present in our approximation: (i) estimating $E_i$ through Monte Carlo sampling, (ii) sub-sampling the data, and (iii) approximating expectation of KL terms with sampling from $q(\alpha)$. Sources of (i) and (ii) are present in the original DSVI approach \citep{Sal}, although in our case we have an extra step in (i) that samples from $q(\alpha)$. Since our inference approach aims to explore the posterior distribution of the new parameter $\alpha$, the extra source of stochasticity (iii) is taken into account for evaluating the ELBO of our DGP. 

The bound is maximized with respect to the variational parameters $\mathbf{m}_n$, $\mathbf{s}_n$, $m_\alpha$, $s_{\alpha}$, inducing related parameters $\mathbf{Z}_n$, $\boldsymbol{\delta}_{\mathbf{Z}_n}$ and model parameters (e.g. covariance function parameters) in each layer. We perform the optimization of the ELBO using a loop procedure consisting of an optimization step with the natural gradient to perform the optimization with respect to the variational parameters of the last layer, then an optimization step using the momentum optimizer ADAM \citep{adam} to perform the optimization for the other parameters in all layers. This optimization procedure has been adopted in DSVI and has shown better results than using only the Adam optimizer for all the layers and parameters \citep{salim}. 

We will compare performance of the resulting DGPs in Section \ref{sec:ill3} for three cases: (i) $\alpha$ is estimated, (ii) $\alpha$ is optimized and (iii) $\alpha=1$. Case (iii) is equivalent to the DGP in \citet{Dunlop} with the choice of $H(f)=\mathrm{exp}(f)$. In this case, the ELBO of the DGP is\\[-20pt]
\begin{equation}
\Lagr_{DGP}=\sum_{i=1}^{n_S}\EE_{q(\{\mathbf{f}_n\}_{n=1}^{N})}\Big(~\textrm{log}~\PP(y_i\vert \mathbf{f}_{N}^i)~\Big)-  \sum_{n=1}^{N}\textrm{KL}\Big(q(\mathbf{u}_n) \parallel \PP(\mathbf{u}_n;\mathbf{Z}_{n},\boldsymbol{\delta}_{\mathbf{Z}_n},\alpha)\Big),\\[-20pt]
\label{eq:dgpelbo2}
\end{equation}
where $$q(\{\mathbf{f}_n\}_{n=1}^{N})=\prod_{n=1}^{N}\Nagr(\mathbf{f}_n \vert \tilde{\boldsymbol{\mu}}_n,\tilde{\mathbf{\Sigma}}_n).$$
and $\alpha=1$. Case (ii) treats $\alpha$ as a parameter to be optimized within \eqref{eq:dgpelbo2}. \\[-40pt]
\subsection{Prediction}\label{sec:pred3} 

Given the variational posterior, we can sample the predictive distribution of $y(.)$ at a new input, $\mathbf{x}^*$. To do this, first we sample $(\epsilon^*_{n})_r\sim \Nagr(0,1)$ and $\hat{\alpha}_r\sim\Nagr(\hat{m}_{\alpha},\hat{s}_{\alpha})$, where $r=1,\dots,R$ and $n=1,\dots,N$ represent the indices of sampling iterations and the number of DGP layers, respectively. Then conditioned on the sampled estimated parameter $\hat{\alpha}_r$ and estimated variational parameters $\hat{\mathbf{Z}}_n$, $\hat{\mathbf{m}}_n$, $\hat{\mathbf{s}}_n$, $\hat{\boldsymbol{\delta}}_{\mathbf{Z}_n}$ obtained through optimizing the ELBO, the sampled variables $(\hat{f}_{1}^*)_r\sim q(\mathbf{f}_{1}\vert\hat{\mathbf{m}}_1,\hat{\mathbf{s}}_1,\mathbf{x}^*,\hat{\mathbf{Z}}_{1})$ and $(\hat{f}_{n}^*)_r\sim q(\mathbf{f}_{n}\vert\hat{\mathbf{m}}_n,\hat{\mathbf{s}}_n,(\hat{f}_{n-1}^*)_r,\mathbf{x}^*,\hat{\mathbf{Z}}_{n},\hat{\boldsymbol{\delta}}_{\mathbf{Z}_n},\hat{\alpha}_r)$ $(n=2,\dots,N)$ are drawn at $\mathbf{x}^*$ recursively 
as \\[-20pt]
\begin{equation*}
(\hat{f}_{1}^*)_r={\mu}_{\hat{\mathbf{m}}_1,\hat{\mathbf{Z}}_1}(\mathbf{x}^*)+ (\epsilon^*_{1})_r \sqrt{{\Sigma}_{\hat{\mathbf{s}}_1,\hat{\mathbf{Z}}_1}(\mathbf{x}^*,\mathbf{x}^*)},
\label{eq:15al}
\end{equation*}
\begin{equation*}
(\hat{f}_{n}^*)_r={\mu}_{\hat{\mathbf{m}}_n,\hat{\mathbf{Z}}_n,\hat{\boldsymbol{\delta}}_{\mathbf{Z}_n},\hat{\alpha}_r}(\mathbf{x}^*,(\hat{f}_{n-1}^*)_r)+ (\epsilon^*_{n})_r \sqrt{{\Sigma}_{\hat{\mathbf{s}}_n,\hat{\mathbf{Z}}_n,\hat{\boldsymbol{\delta}}_{\mathbf{Z}_n},\hat{\alpha}_r}(\mathbf{x}^*,\mathbf{x}^*;(\hat{f}_{n-1}^*)_r)},\\[-10pt]
\label{eq:15all}
\end{equation*}
where $\hat{f}_{n}^*=\hat{f}_n(\mathbf{x}^*)\in\mathbb{R}$ and the posterior predictive means and variances, ${\mu}_{\hat{\mathbf{m}}_1,\hat{\mathbf{Z}}_1}(\mathbf{x}^*)$, ${\mu}_{\hat{\mathbf{m}}_n,\hat{\mathbf{Z}}_n,\hat{\boldsymbol{\delta}}_{\mathbf{Z}_n},\hat{\alpha}_r}(\mathbf{x}^*,(\hat{f}_{n-1}^*)_r)$, and ${\Sigma}_{\hat{\mathbf{s}}_1,\hat{\mathbf{Z}}_1}(\mathbf{x}^*,\mathbf{x}^*)$,  ${\Sigma}_{\hat{\mathbf{s}}_n,\hat{\mathbf{Z}}_n,\hat{\boldsymbol{\delta}}_{\mathbf{Z}_n},\hat{\alpha}_r}(\mathbf{x}^*,\mathbf{x}^*;(\hat{f}_{n-1}^*)_r)$, are given by (\ref{eqn:mean}). The optimized inducing locations $\hat{\mathbf{Z}}_n$ play an analogous role of the design $\mathbf{X}$ in usual kriging formula (\ref{eqn:eq11}), which scale up computations in prediction as well. This sampling procedure is repeated $R$ times to obtain the posterior sample of of $y(.)$ at $\mathbf{x}^*$ along with incorporating all sources of uncertainty from $\alpha$ and hidden layers in the predictive distribution.\\[-40pt]

\section{Illustrative example}\label{sec:ill3}

We compare our emulation methodology on a 2D synthetic example across three cases: (i) $\alpha$ is estimated, (ii) $\alpha$ is optimized and (iii) $\alpha=1$ and is fixed. We fit a DGP with two hidden layers, where the DGP is constructed using the methodology above and using the Matern covariance function with $\nu=2.5$ throughout.
To compare the performance of the DGP, the following criteria are calculated:\\
(1) Nash–Sutcliffe Efficiency (NSE) \\[-25pt]
\begin{equation*}
\label{eq:mspe}
\mathrm{NSE}=1-\frac{\mathrm{MSPE}}{\mathrm{Var}},\\[-10pt]
\end{equation*} 
where Var is the variance of true values $y(.)$ at prediction inputs $\mathbf{x}_1^*,\dots,\mathbf{x}_p^*$ and MSPE represents the mean square prediction error formulated as 
\\[-25pt]
\begin{equation*}
\mathrm{MSPE}=\frac{1}{p}\sum_{i=1}^{p}\Big(\hat{y}(\mathbf{x}_i^*)-y(\mathbf{x}_i^*)\Big)^2,\\[-10pt]
\end{equation*}
where $\hat{y}(\mathbf{x}_i^*)$ is the posterior  predictive mean at the new input $\mathbf{x}_i^*$. Similar to the coefficient of determination, the NSE \citep{nash} attempts to measure the proportion of variation that can be explained by a predictive model. NSE values close to 1 indicate that the emulator has performed well in terms of prediction accuracy.\\
(2) 95$\%$ Coverage Probability (CP) is the proportion of the true values $y(.)$ at prediction inputs that lie within their 95\% credible interval. 

We consider a $2$-d piece-wise computer model 
\begin{equation*}
g(x_1, x_2) = \begin{cases}
1.3 & x_1 \in [0.66,0.91]  \mbox{ and } x_2 \in [0.4,0.91]\\
2.2 & x_1 \in [0.1,0.5]  \mbox{ and } x_2 \in [0.6,0.92]  \\
3.5 & x_1 \in [0.15,0.6] \mbox{ and } x_2 \in [0.1,0.52]\\
0   & \text{o.w}. 
\end{cases} ~~~~~,~~~ x_1,x_2\in[0,1]
\label{eq:2dfunc}
\end{equation*}
shown in panel (a) of Figure \ref{fig:trueg}, which illustrates the type of discontinuity we expect in the COMPAS model. 
\begin{figure}[H]
\centering
\begin{tabular}[c]{cc}
\begin{subfigure}[b]{0.49\textwidth}
    \includegraphics[width=\textwidth]{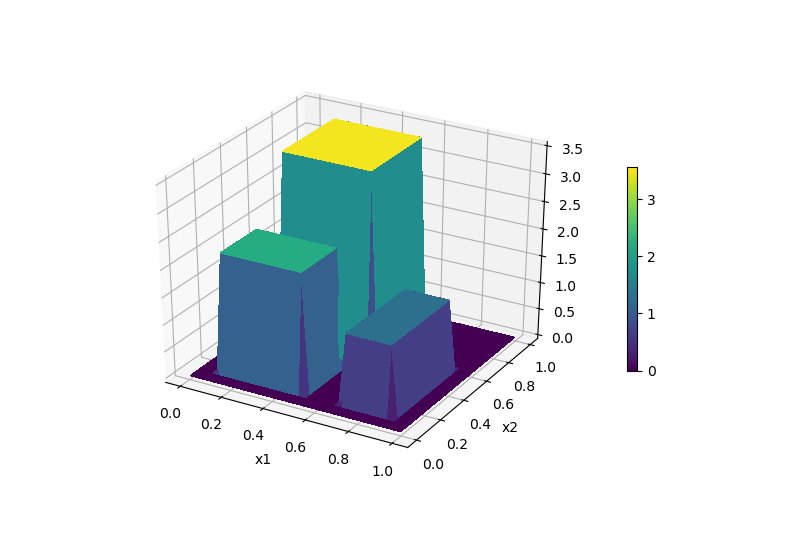}
    \caption{}
    
  \end{subfigure}
\begin{subfigure}[b]{0.32\textwidth}
    \includegraphics[width=\textwidth]{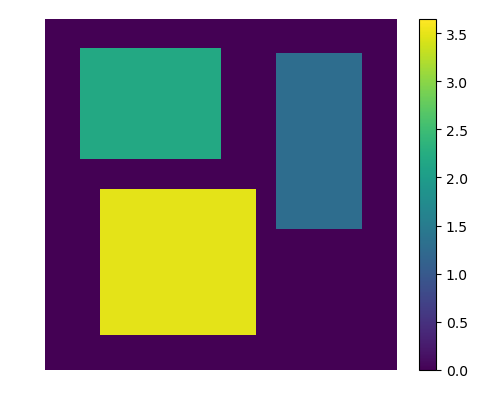}
    \caption{}
    
  \end{subfigure}
\end{tabular} 
\caption{(a) 2-d illustrative computer model with regions of discontinuities (b) Heatmap of true function outputs at prediction points}

  \label{fig:trueg}
\end{figure}
We fit a two hidden layer DGP with the training data evaluated on a $25$ by $25$ grid in $\mathcal{X}=[0,1]^2$, using $m=200$ inducing points in each layer for each of our three cases. In case (i), a normal distribution $\Nagr(3.5,1)$ is chosen for $\PP(\alpha)$ and $q(\alpha)$ is initialized with $\Nagr({m_{\alpha}}_{\mathrm{ini}},{s_{\alpha}}_{\mathrm{ini}})$ where ${m_{\alpha}}_{\mathrm{ini}}=3$ and ${s_{\alpha}}_{\mathrm{ini}}=1$. The response is evaluated over the prediction set, a $70\times70$ grid in $\mathcal{X}=[0,1]^2$, by sampling from the resulting predictive posterior ($5000$ samples) in all cases (i), (ii) and (iii). 

\begin{figure}[h!]
\centering
\begin{tabular}[c]{ccc}
\begin{subfigure}[b]{0.32\textwidth}
    \includegraphics[width=\textwidth]{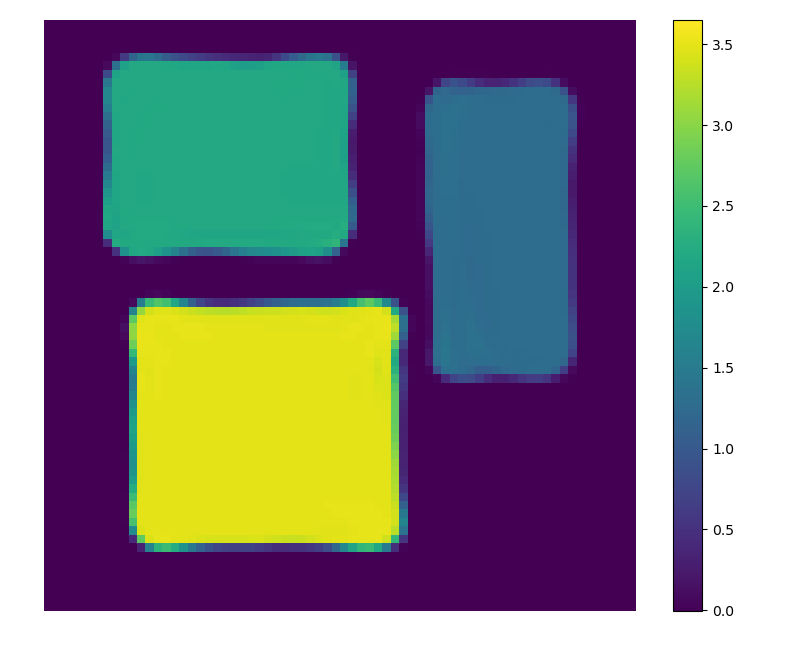}
    \caption{NSE = $91.64\%$}
  \end{subfigure}
\begin{subfigure}[b]{0.32\textwidth}
    \includegraphics[width=\textwidth]{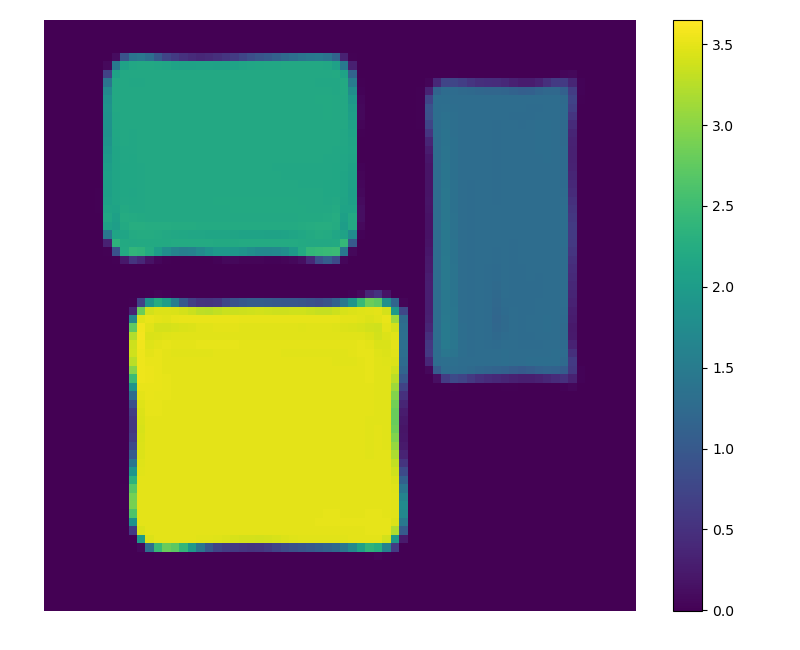}
    \caption{NSE = $91.41\%$}
  \end{subfigure}  
 \begin{subfigure}[b]{0.32\textwidth}
    \includegraphics[width=\textwidth]{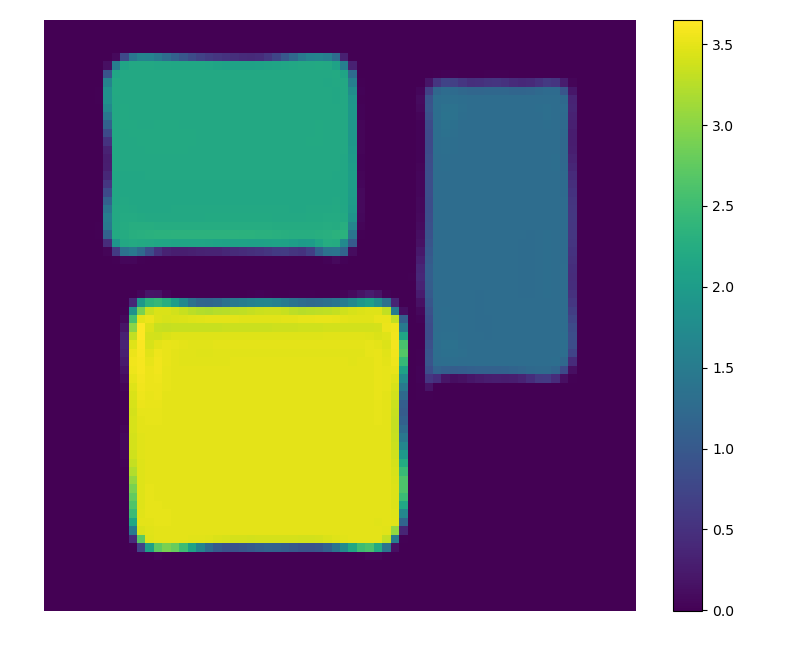}
    \caption{NSE = $88.28\%$}
  \end{subfigure}
\end{tabular} 
\caption{Heatmap of the predictions in case of (a) $\alpha$ is estimated, (b) $\alpha$ is optimized, (c) $\alpha=1$. Plots share the same color bar as given in the left side of each, where brighter colors indicate greater predicted values}
  \label{fig:2dpred}
\end{figure}
\begin{figure}[H]
\centering
\begin{tabular}[c]{ccc}
\begin{subfigure}[b]{0.32\textwidth}
    \includegraphics[width=\textwidth]{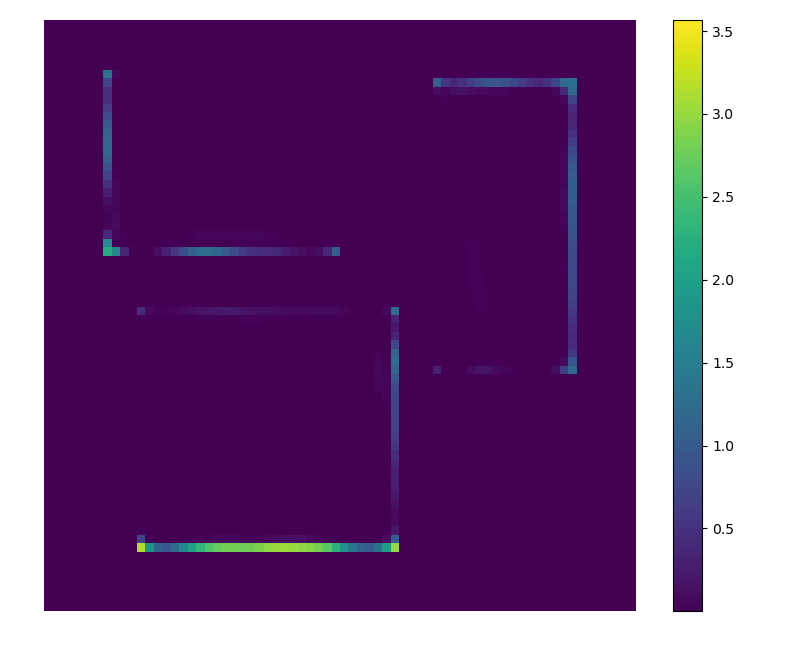}
    \caption{}
  \end{subfigure}
\begin{subfigure}[b]{0.32\textwidth}
    \includegraphics[width=\textwidth]{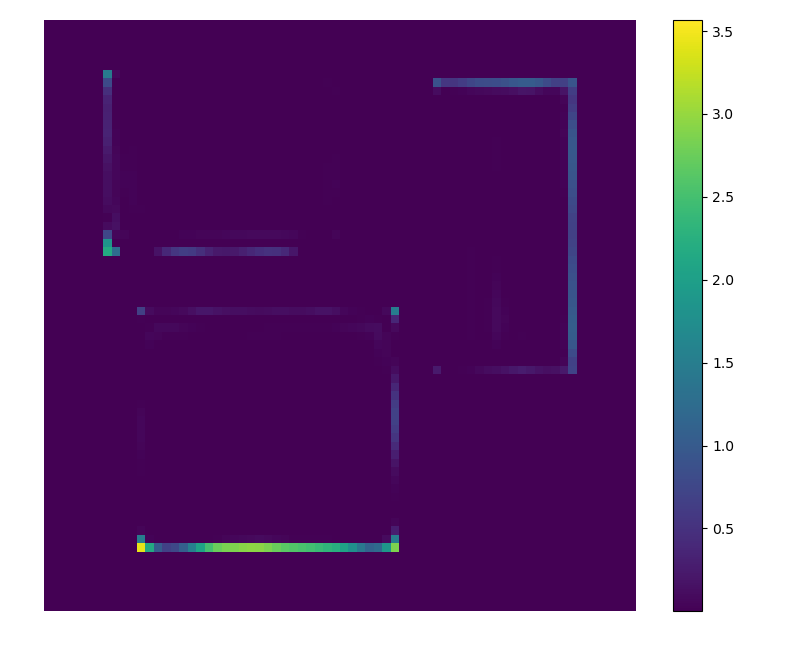}
    \caption{}
  \end{subfigure}  
 \begin{subfigure}[b]{0.32\textwidth}
    \includegraphics[width=\textwidth]{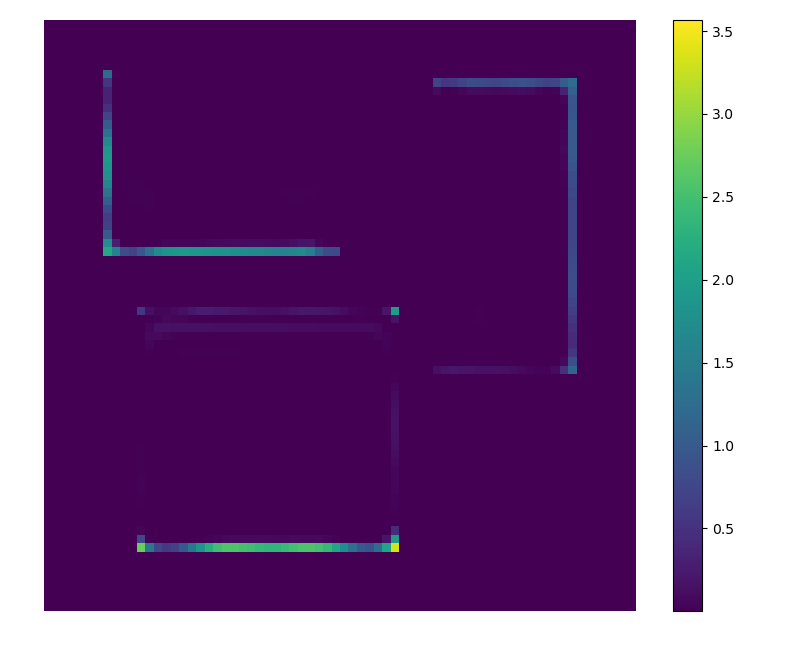}
    \caption{}
  \end{subfigure}
\end{tabular} 
\caption{Heatmap of absolute prediction errors in case of (a) $\alpha$ is estimated, (b) $\alpha$ is optimized, (c) $\alpha=1$. Plots share the same color bar as given in the left side of each, where brighter colors indicate larger errors.
}
 \label{fig:2dabser}
\end{figure}
The true response values at prediction points is shown in a heatmap plot in panel (b) of Figure \ref{fig:trueg}. 
The predictions and absolute prediction errors are compared in panels (a), (b) and (c) of Figures \ref{fig:2dpred} and \ref{fig:2dabser} in cases (i), (ii) and (iii), respectively. In the heatmaps, brighter color corresponds to larger predicted value and larger absolute prediction error. As seen in Figures \ref{fig:2dabser}, in all three cases, larger errors are observed around the boundaries of the three regions, where values of the model response change between zero and positive outputs. In panels (a) and (b) where that parameter $\alpha$ is estimated and optimized, it is evident that the performance of the DGP around the boundaries is greatly improved.
To numerically assess the prediction and uncertainty performance of the DGP model in all cases, $95\%$ CP and NSE are computed and displayed in Table \ref{table:t2d1}. As shown in the table the largest $95\%$ CP in the predictive DGP model is reached in the case of estimating parameter $\alpha$. This is due to including the extra uncertainty from estimating $\alpha$. Also, the DGP model where the posterior distribution for $\alpha$ is estimated ($\hat{q}(\alpha)=\Nagr(\hat{m}_{\alpha},\hat{s}_{\alpha})$ for $\hat{m}_{\alpha}=3.2452$ and $\hat{s}_{\alpha}=0.0308$) explains the largest amount of the response variability, although a similar result is achieved by optimizing $\alpha$ ($\alpha_{opt}=3.3005$).

\begin{table}[h!]
		\caption{Prediction accuracy of the DGP for three different methods}	
		\centering
\begin{tabular}{cccc}
		
			\hline
			 & $\alpha_{est}$ & $\alpha_{opt}$ &$\alpha=1$
			 \\ [0.4ex]
			\hline\\
			CP & $93.82\%$ & $92.69\%$ & $90.27\%$ \\[0.4ex]
		
			$\mbox{NSE}$ & $91.64\%$ & $91.41\%$ & $88.28\%$ \\[0.4ex]\\
			\hline
		\end{tabular}
	
\label{table:t2d1}
\end{table}

\section{Emulating the COMPAS Model}\label{num:com}

The proposed methodology is motivated by the need to emulate the COMPAS model. The input and the output variables of the COMPAS model are displayed in Table \ref{table:comd}. The inputs that were varied correspond to the initial conditions of the binary star system (see \citet{li} for details).
The output of the model is the chirp mass of a BBH merger. If no BBH merger is observed, the output would be ``NA'', although in our work, the chirp mass is considered to be zero.

\begin{table}[h!]
		\caption{Input and output of COMPAS model}
		\begin{center}\footnotesize{
\begin{tabular}{lc}
			\hline
			\textbf{Input}& \textbf{Range} \\[0.4ex]%
			\hline\\
			\textbf{Initial conditions:} $\mathbf{x}$ & \\[0.4ex]\\
			$m_1$ : the mass of the initially more massive star & [8,150]$M_{\odot}$ \\[0.4ex]
			$m_2$ : the mass of the initially less massive star&(0.1 $M_{\odot}$, $m_1$]\\[0.4ex]
			a: the initial orbital separation & [0.01,1000]AU \\[0.4ex]
			v$_i$ : supernova natal kick vector for supernova $i$ &\\[0.4ex]
			for $i = 1, 2$ including :  &\\[0.4ex]
			~~$v_i$ - magnitude of the supernova natal kick
            & [0,$\infty$)\\[0.4ex]
            ~~(km $s^{-1}$)&\\[0.4ex]
			~~${\theta}_i$ - polar angle defining the direction of the natal& [0,$\pi$]\\[0.4ex]
			~~~~~~~kick &\\[0.4ex]
			~~${\phi}_i$ - azimuthal angle defining the direction of the & [0,$2\pi$]\\[0.4ex]
			~~~~~~~natal kick & \\[0.4ex]
			~~${\omega}_i$ - mean anomaly & [0,$2\pi$]\\[0.4ex]\\
			\hline
			\textbf{Output}&\\[0.4ex]%
			\hline\\
			${\mathcal{M}_c}$: chirp mass of BBH & (0,150)$M_{\odot}$ or NA\\[0.4ex]\\
			\hline
		\end{tabular}}		
\label{table:comd}
\end{center}
\end{table}

The challenges in emulation of COMPAS model (e.g., see  Section \ref{sec:app}) are the regions of discontinuities in the response surface of the chirp mass and the large number of simulation runs with very low success rate for BBH formation. We apply our proposed emulation method to two million computer model runs where roughly $24\%$ of the simulations resulted in a chirp mass output. In the remaining simulations, no BBH was formed and thus no chirp mass is computed. 
\begin{figure}[h!]
\centering
\begin{tabular}[c]{c}
\begin{subfigure}[b]{1.0\textwidth}
    \includegraphics[width=\textwidth]{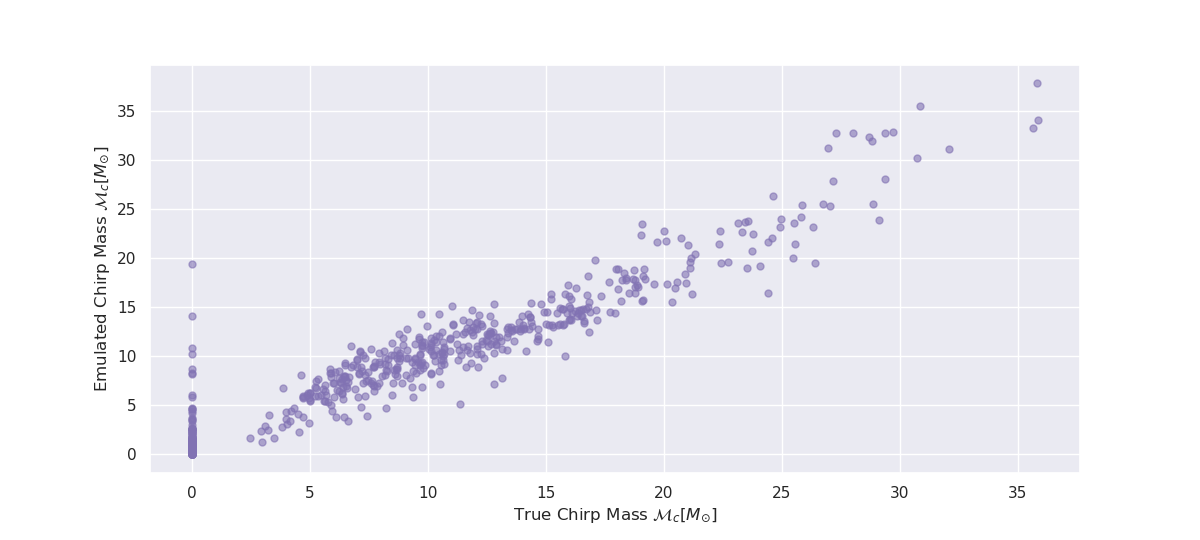}
  \end{subfigure}

\end{tabular} 

\centering
\begin{tabular}[c]{c}
\begin{subfigure}[b]{1.0\textwidth}
    \includegraphics[width=\textwidth]{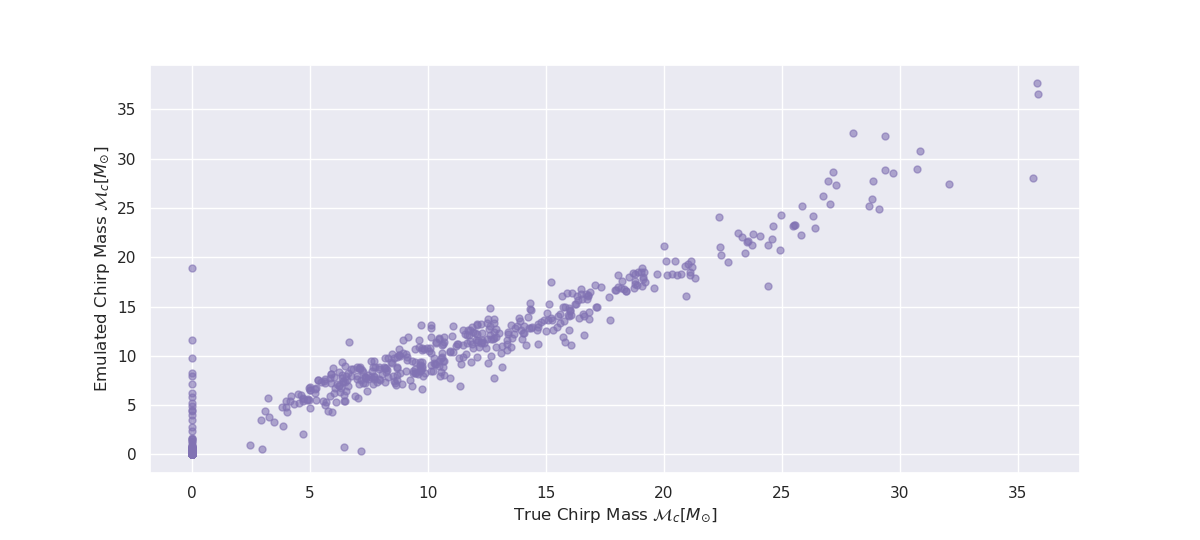}
  \end{subfigure}
\end{tabular} 
\caption{Emulated chirp mass against the true chirp mass for the DGP with two hidden layers (upper) and the DGP with three hidden layers (lower).}

\label{fig:emutvsp}
\end{figure}
We standardized input variables of the simulation data to the 11-dimensional unit hypercube $[0, 1]^{11}$. A random prediction set with size of $1000$ including 450 successful simulations (active points) was held out from the data to evaluate performance of our DGP emulator and the remaining simulations were used for constructing the emulator. 

We fit the DGP with two and three layers with $100$ inducing points in each layer.The DGPs are constructed using 
Matern covariance functions with $\nu=2.5$. For training, we approximated the ELBO (\ref{eqn:ev2}) with a batch size of $10,000$ to achieve scalability. Parameters $m_\alpha$ and $s_{\alpha}$ were found by maximizing the ELBO to obtain an estimate of the variational posterior distribution of $\alpha$ (i.e. $\hat{q}(\alpha)=\Nagr(\hat{m}_{\alpha},\hat{s}_{\alpha})$). Emulation of the response is done for the the prediction set with $5000$ samples from the resulting predictive posterior distribution. 

Figure \ref{fig:emutvsp} shows the emulated chirp mass against the true chirp mass at 1000 prediction points using the DGP with two and three hidden layers. As seen in the figure, most of the active and non-active points lie at or near the 45 degree line, meaning that the proposed method appears successful at emulating the response. The largest errors appear where the true response was zero, but the emulator predicts a non-zero value and vice versa. Comparing two and three hidden layer DGPs in Figure \ref{fig:emutvsp}, we see that points are more tightly centred around the 45 degree line for the three hidden layer model (lower plot).

\begin{figure}[h!]
\centering
\begin{tabular}[c]{c}
 \begin{subfigure}[b]{1.0\textwidth}
    \includegraphics[width=\textwidth]{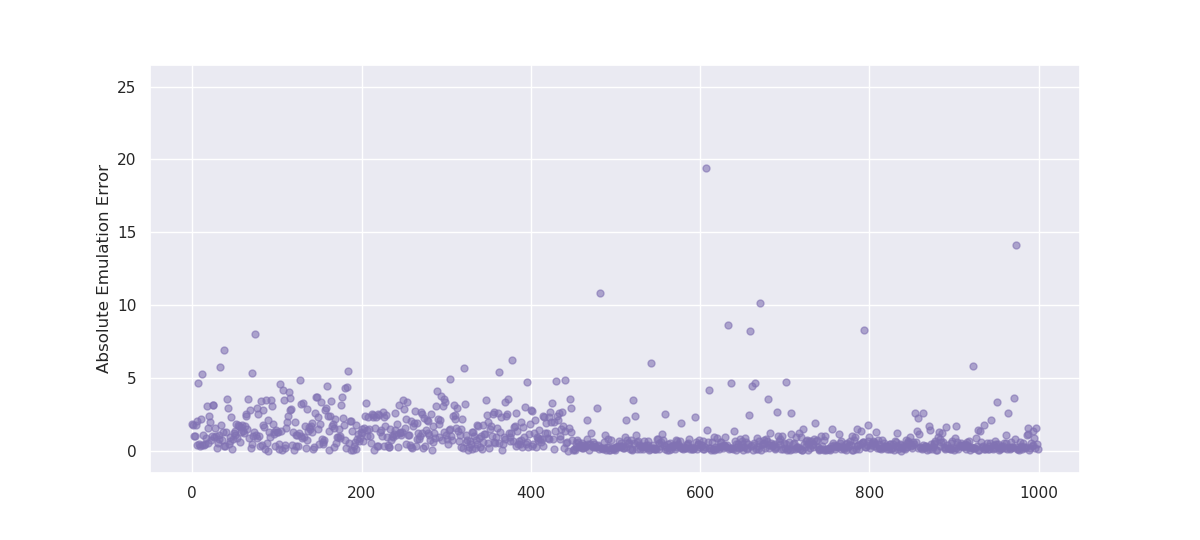}
  \end{subfigure}
\end{tabular} 

\centering
\begin{tabular}[c]{c}
 \begin{subfigure}[b]{1.0\textwidth}
    \includegraphics[width=\textwidth]{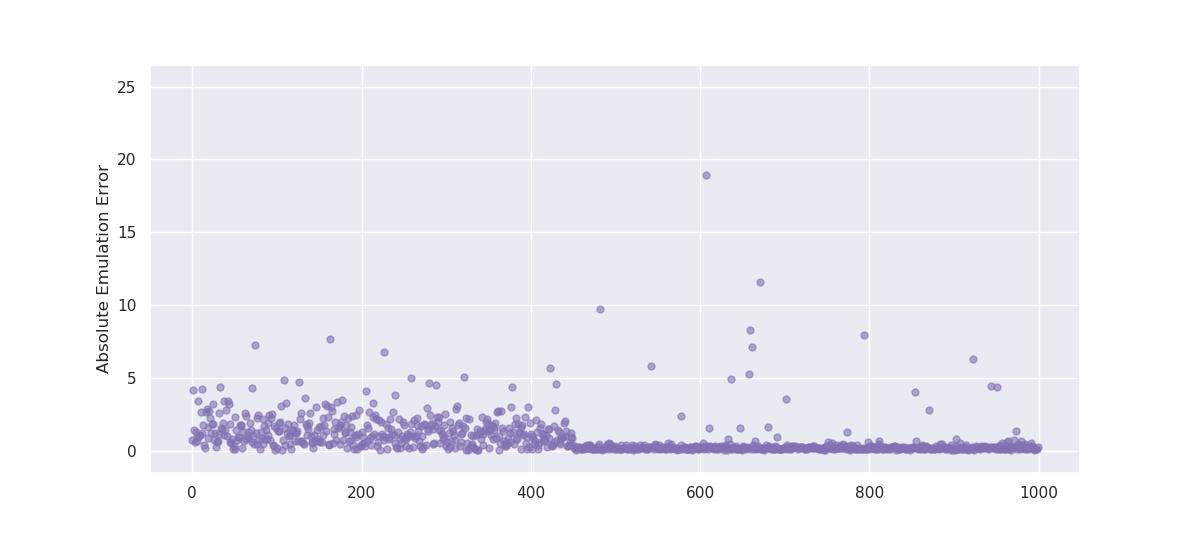}
  \end{subfigure}
\end{tabular} 
\caption{Absolute emulation errors for the DGP with two hidden layers (upper) and the DGP with three hidden layers (lower).}
\label{fig:emuabs}
\end{figure}

The absolute emulation errors are plotted in Figure \ref{fig:emuabs}. Looking at the figure, most of the active and non-active points have small errors close to zero, although there are large errors for a few predictions for reasons explained for previous figure. Also we also see from the the lower plot that there are more small errors than in the upper plot, meaning that the three layer DGP reduces the emulation error compared with the two layer DGP.

Table \ref{table:comt} displays the performance criteria including NSE and $95\%$ CP computed for the DGP emulator with two and three hidden layers. As it is expected, the predictive model explains a larger proportion of the response variability (NSE=$95.76\%$) in the three layer DGP, although its $95\%$ CP is less than the two layer DGP. The estimated parameters $\hat{m}_{\alpha}$ and $\hat{s}_{\alpha}$ displayed in the table show that smoothness parameter $\alpha$ is estimated around $0.7$. We attribute it to the size of the data and the higher dimensional input space of the COMPAS model. The training time of this large data set shown in the table illustrate how our emulation method can be computationally efficient. Also it shows that increasing number of DGP layers from two to three increases the computational time.
\vspace{5mm}
\begin{table}[h!]
		\caption{COMPAS emulation results using two and three hidden layer DGP }
		\centering
		\begin{tabular}{c c c c c c}
			\hline\hline
			\# of Layer & $\hat{m}_{\alpha}$& $\hat{s}_{\alpha}$ & NSE & CP & Training Time (1 iter)  \\ [0.1ex] 
			\hline
			2&0.7484& 0.0003& 94.31\%& 91.7 &0.31 s \\[0.1ex]
			3&0.7695& 0.0002& 95.76\%& 90.4 &0.35 s\\[0.1ex] 
			\hline
		\end{tabular}
		\label{table:comt}

	\end{table}   
\\[-40pt]
\section{Summary and Discussion}\label{sec:sum3}

In this paper, the DGP of \citet{Dunlop} was investigated and modified. This work was motivated by the application of emulation of COMPAS model, a binary population synthesis model that simulates the formation of binary black holes (BBHs). We proposed a non-stationary DGP emulator which can be adapted to the type of discontinuities found on the COMPAS response surface, and scales to the large number of simulation runs. Our approach can be applied 
to a large class of complex computer models, and scales to arbitrarily large simulation designs as well. 

We generalized the notation of two broad forms of DGP \citep{law,Dunlop} to emphasize their differences and properties. Specifically, we addressed exactly what the DGP forms would look like if the data were actually a realization of a stationary GP in two propositions. Some advantages of the DGP of \citet{Dunlop} such as operating covariance functions on the original input space rather than the warped input space, led us to consider it in our framework. Our DGP emulator was proposed by modifying the \citet{Dunlop} formulation through introducing a new parameter that controls smoothness of the DGP layers. The impact of the proposed parameter on the level of smoothness of the hidden layers was theoretically illustrated and numerically visualised through DGP realizations. Particularly, in Proposition \ref{pro:alpha}, we proved that what is happening in the non-stationary covariance function as the proposed parameter $\alpha$ changes from zero to infinity.

Doing inference on the DGP emulator using sampling methods such as MCMC is computationally infeasible in large designs. Hence, we employed a VI approach to overcome the computational issues and adapted it to be able to estimate the posterior distribution of the smoothness parameter $\alpha$. Specifically, we adapted the DVSI method introduced by \citet{Sal} (i) to be suitable for the DGP model of \citet{Dunlop} modified by our proposed parameter $\alpha$ and (ii) to our framework, emulation of deterministic computer models. Our modified approach allowed us to explore the posterior distribution of the smoothness of the model response surface and was demonstrated to preserve accuracy with uncertainty measures for arbitrary large designs. 

Further possible innovations to the DGP model of \citet{Dunlop} could be through (i) specifying an $\alpha$ for each coordinate dimension of the input space and/or at each non-stationary layer separately, and (ii) changing the dimensionality between layers. By combining methods (i) and (ii) new different variants of the DGP can be proposed. Although, these new DGP variants get the benefit of having multidimensional layers in the DGP from (\ref{eq:44}) and explicitly non-stationary covariance functions in the DGP model through (\ref{eq:66}), more complexities are introduced into the original model (\ref{eq:44}) by adding extra unobserved latent variables in each dimension and adding extra parameters to be estimated. We defer a more thorough investigation of these new variants to future work. \\[-35pt]

\section{Supplementary Materials}
The supplemental materials contain (1) proofs of Propositions \ref{eq:lfix}, \ref{eq:dfix}, \ref{pro:alpha} and Corollary \ref{eq:cor}, and (2) the derivations of marginalizing inducing variables in (\ref{eqn:mar1}) and the ELBO of our DGP in \ref{eq:dgpelbo}. \\[-35pt]

\bibliographystyle{plainnat}

\begin{thebibliography}{28}
\providecommand{\natexlab}[1]{#1}
\providecommand{\url}[1]{\texttt{#1}}
\expandafter\ifx\csname urlstyle\endcsname\relax
  \providecommand{\doi}[1]{doi: #1}\else
  \providecommand{\doi}{doi: \begingroup \urlstyle{rm}\Url}\fi

\bibitem[Barrett et~al.(2018)Barrett, Gaebel, Neijssel, VignaGomez, Stevenson, Berry, Farr, and Mandel]{Bar}
J.W. Barrett, S.M. Gaebel, C.J. Neijssel, A.~VignaGomez, S.~Stevenson, C.P.L. Berry, W.M. Farr, and I.~Mandel.
\newblock Accuracy of inference on the physics of binary evolution from gravitational-wave observations.
\newblock \emph{MNRAS}, 477\penalty0 (4):\penalty0 4685--4695, 2018.

\bibitem[Broekgaarden et~al.(2019)Broekgaarden, Justham, DeMinK, Gair, Mandel, Stevenson, Barrett, VignaGomez, and Neijssel]{broe}
F.S. Broekgaarden, S.~Justham, S.E. DeMinK, J.~Gair, I.~Mandel, S.~Stevenson, J.W. Barrett, A.~VignaGomez, and C.J. Neijssel.
\newblock Stroop-wafel: Simulating rare outcomes from astrophysical populations, with application to gravitational-wave sources.
\newblock \emph{Monthly Notices of the Royal Astronomical Society}, 490:\penalty0 5228--5248, 2019.

\bibitem[Damianou and Lawrence(2013)]{law}
A.~Damianou and N.~Lawrence.
\newblock Deep gaussian processes.
\newblock \emph{In Artificial Intelligence and Statistics, PMLR}, 31:\penalty0 207–215, 2013.

\bibitem[Dunlop et~al.(2018)Dunlop, Girolami, Stuart, and Teckentrup]{Dunlop}
M.M. Dunlop, M.A. Girolami, A.M. Stuart, and A.L. Teckentrup.
\newblock How deep are deep gaussian processes?
\newblock \emph{Journal of Machine Learning Research}, 19:\penalty0 1--46, 2018.

\bibitem[Harari et~al.(2018)Harari, Bingham, Dean, and Higdon]{har}
O.~Harari, D.~Bingham, A.~Dean, and D.~Higdon.
\newblock Computer experiments: prediction accuracy, sample size and model complexity revisited.
\newblock \emph{Statistica Sinica}, 28\penalty0 (2):\penalty0 899–919, 2018.

\bibitem[Higdon et~al.(1999)Higdon, Swall, and Kern]{hig2}
D.~Higdon, J.~Swall, and J.~Kern.
\newblock Non-stationary spatial modeling.
\newblock \emph{Bayesian Statistics}, 6\penalty0 (1):\penalty0 761–768, 1999.

\bibitem[Jones et~al.(1998)Jones, Schonlau, and Welch]{jones}
D.R. Jones, M.~Schonlau, and W.J. Welch.
\newblock Efficient global optimization of expensive black-box functions.
\newblock \emph{Journal of Global Optimization}, 13\penalty0 (4):\penalty0 455–492, 1998.

\bibitem[Kaufman et~al.(2011)Kaufman, Bingham, Habib, Heitmann, and Frieman]{Kauf}
C.G. Kaufman, D.~Bingham, S.~Habib, K.~Heitmann, and J.A. Frieman.
\newblock Efficient emulators of computer experiments using compactly supported correlation functions, with an application to cosmology.
\newblock \emph{The Annals of Applied Statistics}, 4\penalty0 (5):\penalty0 2470–2492, 2011.

\bibitem[Kingma and Ba(2014)]{adam}
D.P. Kingma and J.L. Ba.
\newblock Adam: A method for stochastic optimization.
\newblock 2014.
\newblock \doi{1412.6980}.

\bibitem[Kingma et~al.(2015)Kingma, Salimans, and Welling]{kin}
D.P. Kingma, T.~Salimans, and M.~Welling.
\newblock Variational dropout and the local reparameterization trick.
\newblock \emph{Advances in Neural Information Processing Systems}, 28, 2015.
\newblock \doi{1506.02557}.

\bibitem[Lin et~al.(2021)Lin, Bingham, Broekgaarden, and Mandel]{li}
L.~Lin, D.~Bingham, F.~Broekgaarden, and I.~Mandel.
\newblock Uncertainty quantification of a computer model for binary black hole formation.
\newblock \emph{The Annals of Applied Statistics}, 15\penalty0 (4):\penalty0 1604--1627, 2021.

\bibitem[Mandel and Farmer(2018)]{Man}
I.~Mandel and A.~Farmer.
\newblock Merging stellar-mass binary black holes.
\newblock 2018.
\newblock \doi{1806.05820}.

\bibitem[Ming and Williamson(2023)]{min2}
D.~Ming and D.~Williamson.
\newblock Linked deep gaussian process emulation for model networks.
\newblock 2023.
\newblock \doi{10.48550/arXiv.2306.01212}.

\bibitem[Ming et~al.(2021)Ming, Williamson, and Guillas]{dw}
D.~Ming, D.~Williamson, and S.~Guillas.
\newblock Deep gaussian process emulation using stochastic imputation.
\newblock 2021.
\newblock \doi{2107.01590}.

\bibitem[Moriarty-Osborne and Teckentrup(2023)]{con}
C.~Moriarty-Osborne and A.~L. Teckentrup.
\newblock Convergence rates of non-stationary and deep gaussian process regression.
\newblock \emph{arXiv preprint arXiv:2312.07320}, 2023.

\bibitem[Nash and Sutcliffe(1970)]{nash}
J.E. Nash and J.V. Sutcliffe.
\newblock River flow forecasting through conceptual models part i - a discussion of principles.
\newblock \emph{Journal of Hydrology}, 10:\penalty0 282--290, 1970.

\bibitem[Paciorek and Schervish(2004)]{Paciorek}
C.~Paciorek and M.~Schervish.
\newblock Non-stationary covariance functions for gaussian process regression.
\newblock \emph{Advances in Neural Information Processing Systems}, 16:\penalty0 273–280, 2004.

\bibitem[Radaideh and Kozlowski(2020)]{rada}
M.I. Radaideh and T.~Kozlowski.
\newblock Surrogate modeling of advanced computer simulations using deep gaussian processes.
\newblock \emph{Reliability Engineering and System Safety}, 195:\penalty0 106731, 2020.

\bibitem[Rajaram et~al.(2020)Rajaram, Puranik, Ashwin~Renganathan, Sung, Fischer, Mavris, and Ramamurthy]{raja}
D.~Rajaram, T.G. Puranik, S.~Ashwin~Renganathan, W.~Sung, O.P. Fischer, D.N. Mavris, and A.~Ramamurthy.
\newblock Empirical assessment of deep gaussian process surrogate models for engineering problems.
\newblock \emph{Journal of Aircraft}, page 1–15, 2020.

\bibitem[Sacks et~al.(1989)Sacks, Welch, Mitchell, and Wynn]{Sacks}
J.~Sacks, W.J. Welch, T.J. Mitchell, and H.P. Wynn.
\newblock Design and analysis of computer experiments.
\newblock \emph{Statist. Sci.MR1041765}, 4:\penalty0 409–435, 1989.

\bibitem[Salimbeni(2020)]{salim}
H.~Salimbeni.
\newblock Deep gaussian processes: Advances in models and inference.
\newblock \emph{PhD thesis, Imperial College London}, 2020.

\bibitem[Salimbeni and Deisenroth(2017)]{Sal}
H.~Salimbeni and M.~Deisenroth.
\newblock Doubly stochastic variational inference for deep gaussian processes.
\newblock \emph{In Advances in Neural Information Processing Systems}, page 44588–4599, 2017.

\bibitem[Sampson and Guttorp(1992)]{Samp}
P.D. Sampson and P.~Guttorp.
\newblock Nonparametric estimation of nonstationary spatial covariance structure.
\newblock \emph{Journal of the American Statistical Association}, 87\penalty0 (417):\penalty0 108--119, 1992.

\bibitem[Sauer et~al.(2020)Sauer, Gramacy, and Higdon]{ann}
A.~Sauer, R.B. Gramacy, and D.~Higdon.
\newblock Active learning for deep gaussian process surrogates.
\newblock 2020.
\newblock \doi{2012.08015}.

\bibitem[Sauer et~al.(2022)Sauer, A., and Gramacy]{sau2}
A.~Sauer, Cooper A., and R.B. Gramacy.
\newblock Vecchia-approximated deep gaussian processes for computer experiments.
\newblock 2022.
\newblock \doi{10.48550/arXiv.2204.02904}.

\bibitem[Schmidt and O’Hagan(2003)]{Smit}
A.M. Schmidt and A.~O’Hagan.
\newblock Bayesian inference for nonstationary spatial covariance structure via spatial deformations.
\newblock \emph{Journal of the Royal Statistical Society, Series B}, 65:\penalty0 745–758, 2003.

\bibitem[Yazdi(2022)]{yazdi}
F.~Yazdi.
\newblock Fast deep gaussian process modeling and design for large complex computer experiments.
\newblock \emph{PhD thesis, Simon Fraser University}, 2022.

\bibitem[Zilber and Katzfuss(2021)]{kat}
D.~Zilber and M.~Katzfuss.
\newblock Vecchia–laplace approximations of generalized gaussian processes for big non-gaussian spatial data.
\newblock \emph{Computational Statistics \& Data Analysis}, 153, 2021.
\newblock ISSN 0167-9473.
\newblock \doi{1https://doi.org/10.1016/j.csda.2020.107081}.

\end{thebibliography}

\appendix

\section{Supplemental Material for Deep Gaussian Process Emulation and Uncertainty Quantification for Large Complex Computer Experiments}

~~~Proof for Proposition \ref{eq:lfix}\\[-40pt]
\begin{proof}
$\Rightarrow$) Let $f_{n,l}(\mathbf{x})\vert f_{n-1}(\mathbf{x})$ be a stationary GP. This implies that  $k_{n,l}(f_{n-1}(\mathbf{x});\bm{\phi}_{n})$ is a stationary covariance function, where the input space is not warped through $f_{n-1}(.)$. Hence $f_{n-1}(\mathbf{x})$ must be a linear transformation of $\mathbf{x}$. i.e there exists a vector of constants $\mathbf{c}\in\mathbb{R}^d$ such that $f_{n-1}(\mathbf{x})=\mathbf{c}\odot\mathbf{x}$ for any $\mathbf{x}\in\mathcal{X}\subseteq\mathbb{R}^d$.

$\Leftarrow$) Let $f_{n-1}(\mathbf{x})=\mathbf{c}\odot\mathbf{x}$ for a vector of constants $\mathbf{c}\in\mathbb{R}^d$. This implies $k_{n,l}(\mathbf{c}\odot\mathbf{x};\bm{\phi}_{n})$ where $\mathbf{c}$ scales the length scale vector in $\bm{\phi}_{n}$. As a result the covariance function $k_{n,l}(.,.)$ operates on the input space. So, $f_{n,l}(\mathbf{x})\vert f_{n-1}(\mathbf{x})$ is a stationary GP.
\end{proof}

Proof for Proposition \ref{eq:dfix}\\[-40pt]
\begin{proof}
$\Rightarrow$) Let $f_n(\mathbf{x})\vert f_{n-1}(\mathbf{x})$ be a stationary GP. This implies that $k_{n}(\mathbf{x};\bm{\phi}_{n}(f_{n-1}(\mathbf{x})))$ is a stationary covariace function. Hence covariance parameter $\bm{\phi}_{n}(f_{n-1}(\mathbf{x})))$ must be a constant and would not change respect to any space location $\mathbf{x}\in\mathcal{X}\subseteq\mathbb{R}^d$. To have this, $f_{n-1}(\mathbf{x})$ must be a constant for any $\mathbf{x}\in\mathcal{X}\subseteq\mathbb{R}^d$.

$\Leftarrow$) Let $f_{n-1}(\mathbf{x})$ be a constant. This implies that covariance parameter $\bm{\phi}_{n}(f_{n-1}(\mathbf{x})))$ in $k_{n}(\mathbf{x};\bm{\phi}_{n}(f_{n-1}(\mathbf{x})))$ is a constant too and would not change respect to the space location $\mathbf{x}$. As a result $k_n(.,.)$ would be a stationary covariance function and $f_n(\mathbf{x})\vert f_{n-1}(\mathbf{x})$ is a stationary GP.
\end{proof}

Proof for Corollary \ref{eq:cor}\\[-40pt]
\begin{proof}
$\Rightarrow$) Let $y(\mathbf{x})\sim \gp(\mathbf{0}, k(\mathbf{x};\bm{\phi}))$ where  $k(.,.)$ is a stationary covariance function with a fixed parameter $\bm{\phi}$. If we define $f_1(\mathbf{x})=\bm{\phi}$ for any $\mathbf{x}\in\mathcal{X}\subseteq\mathbb{R}^d$, then $y(\mathbf{x})\vert f_{1}(\mathbf{x})\sim \gp(\mathbf{0}, k(\mathbf{x};f_1(\mathbf{x})))$. Hence under the definition in (\ref{eq:66}), $y(.)$ is a DGP with one constant hidden layer.

$\Leftarrow$) It is clear by proposition \ref{eq:dfix}. \end{proof}

Proof for Proposition \ref{pro:alpha}\\[-40pt]
\begin{proof}
Let $\alpha=0$. Then the length scale function $H(f_{n-1}(\mathbf{x}))=\mathrm{exp}(\alpha f_{n-1}(\mathbf{x}))=1$ for all $\mathbf{x}\in\mathcal{X}\subseteq\mathbb{R}^d$, i.e. length scale function does not change through the input space $\mathcal{X}$ in each layer. It follows that the prefactor in (\ref{eq:paci2}) equals $1$ and the square root of the quadratic form in (\ref{eq:rootsq}) becomes
$\sqrt{Q(\mathbf{x},\mathbf{x}^\prime)}=\parallel\mathbf{x}-\mathbf{x}^\prime\parallel_2$. As a result, for $n>1$  $k_n(\mathbf{x},\mathbf{x}';f_{n-1}(\mathbf{x}),f_{n-1}(\mathbf{x}'))=\sigma_n^2~\rho(\parallel\mathbf{x}-\mathbf{x}^\prime\parallel_2)$, a stationary covariance function. This is where we reach stationarity in the DGP. 
As $\alpha$ increases, the degree of the smoothness of the DGP layers get smaller. To show this, we investigate correlation between any two outputs at $\mathbf{x},\mathbf{x}'\in\mathcal{X}\subseteq\mathbb{R}^d$ in three situations. 

Let $\alpha\rightarrow\infty$. (i) If $f_{n-1}(\mathbf{x}), f_{n-1}(\mathbf{x}')>0$, then $H(f_{n-1}(\mathbf{x}))\rightarrow\infty$ and $H(f_{n-1}(\mathbf{x}'))\rightarrow\infty$ and the quadratic form $Q(\mathbf{x},\mathbf{x}')$ in (\ref{eq:rootsq}) goes to zero. Since $\rho(0)=1$, then $\rho(\sqrt{Q(\mathbf{x},\mathbf{x}')})\rightarrow1$. Furthermore, the perfactor in (\ref{eq:paci2}) can be written as \\[-15pt]
\begin{equation}
2^{d/2} \Bigg[\frac{\mathrm{exp}(\alpha f(\mathbf{x}))~\mathrm{exp}(\alpha f(\mathbf{x^\prime}))}{(\mathrm{exp}(\alpha f(\mathbf{x}))+\mathrm{exp}(\alpha f(\mathbf{x^\prime})))^2} \Bigg]^{d/4}.
\label{eq:mpre}
\end{equation} 
Since $\mathrm{exp}(.)$ is a non-negative function, the denominator in (\ref{eq:mpre}) goes to infinity faster than the numerator and the prefactor goes to zero. This implies that $k_n(\mathbf{x},\mathbf{x}';f_{n-1}(\mathbf{x}),f_{n-1}(\mathbf{x}'))\rightarrow 0$ at any inputs $\mathbf{x}$ and $\mathbf{x}'$, thereby having less smooth realizations. (ii) If $f_{n-1}(\mathbf{x})>0, f_{n-1}(\mathbf{x}')<0$, similar to the case (i), $k_n(\mathbf{x},\mathbf{x}';f_{n-1}(\mathbf{x}),f_{n-1}(\mathbf{x}'))\rightarrow 0$. (iii) If $f_{n-1}(\mathbf{x}), f_{n-1}(\mathbf{x}')$ $<0$, then $H(f_{n-1}(\mathbf{x}))\rightarrow 0$ and $H(f_{n-1}(\mathbf{x}'))\rightarrow 0$ and as a result the quadratic form $Q(\mathbf{x},\mathbf{x}')$ in (\ref{eq:rootsq}) goes to $\infty$. Since $\mathrm{lim}_{r\rightarrow\infty}\rho(r)=0$ (by Proposition 2 in \citet{Dunlop}), then $\rho(\sqrt{Q(\mathbf{x},\mathbf{x}')})\rightarrow 0$. Hence $k_n(\mathbf{x},\mathbf{x}';f_{n-1}(\mathbf{x}),f_{n-1}(\mathbf{x}'))\rightarrow 0$ at any inputs $\mathbf{x}$ and $\mathbf{x}'$, thereby having rougher realizations again. 
\end{proof}

Derivation for marginalizing inducing variables in (\ref{eqn:mar1}) as \\[-25pt] $$\int \PP(\mathbf{f}_n\vert\mathbf{u}_n;\mathbf{f}_{n-1},\mathbf{X},\mathbf{Z}_{n},\boldsymbol{\delta}_{\mathbf{Z}_n},\alpha)q(\mathbf{u}_n)~d\mathbf{u_n}=\Nagr(\mathbf{f}_n \vert \tilde{\boldsymbol{\mu}}_n,\tilde{\mathbf{\Sigma}}_n),\\[-10pt]$$ for $n=2,\dots,N$ where \\[-30pt]
$$[\tilde{\boldsymbol{\mu}}_n]_i=\boldsymbol{\Gamma}_n(\mathbf{x}_i,\mathbf{f}_{n-1}^i)^T\mathbf{m}_n,\\[-15pt]$$
$$[\tilde{\mathbf{\Sigma}}_n]_{ij}=k_n(\mathbf{x}_i,\mathbf{x}_j;\mathbf{f}_{n-1}^i,\mathbf{f}_{n-1}^j,\alpha)-\boldsymbol{\Gamma}_n(\mathbf{x}_i,\mathbf{f}_{n-1}^i)^T\Big[k_n(\mathbf{Z}_n,\mathbf{Z}_n;\boldsymbol{\delta}_{\mathbf{Z}_n},\alpha)-\mathbf{s}_n\Big]\boldsymbol{\Gamma}_n(\mathbf{x}_j,\mathbf{f}_{n-1}^j),\\[-15pt]$$
$$\boldsymbol{\Gamma}_n(\mathbf{x}_i,\mathbf{f}_{n-1}^i)=k_n(\mathbf{Z}_n,\mathbf{Z}_n;\boldsymbol{\delta}_{\mathbf{Z}_n},\alpha)^{-1}k_n(\mathbf{Z}_n,\mathbf{x}_i;\boldsymbol{\delta}_{\mathbf{Z}_n},\mathbf{f}_{n-1}^i,\alpha).$$
\begin{proof} 
We have $\PP(\mathbf{f}_n\vert\mathbf{u}_n;\mathbf{f}_{n-1},\mathbf{X},\mathbf{Z}_{n},\boldsymbol{\delta}_{\mathbf{Z}_n},\alpha)=\Nagr(\mathbf{f}_n \vert \boldsymbol{\mu}_n,\mathbf{\Sigma}_n)$, where \\[-25pt]
$$[{\boldsymbol{\mu}_n}]_i=\boldsymbol{\Gamma}_n(\mathbf{x}_i,\mathbf{f}_{n-1}^i)^T\mathbf{u}_n,\\[-15pt]$$
$$[\mathbf{\Sigma}_n]_{ij}=k_n(\mathbf{x}_i,\mathbf{x}_j;\mathbf{f}_{n-1}^i,\mathbf{f}_{n-1}^j,\alpha)-\boldsymbol{\Gamma}_n(\mathbf{x}_i,\mathbf{f}_{n-1}^i)^T k_n(\mathbf{Z}_n,\mathbf{Z}_n;\boldsymbol{\delta}_{\mathbf{Z}_n},\alpha)\boldsymbol{\Gamma}_n(\mathbf{x}_j,\mathbf{f}_{n-1}^j).\\[-10pt]$$
For simplifying the notations we assume that \\[-25pt]
$$\boldsymbol{\mu}_n=\mathbf{\Sigma}_{\mathbf{f}_n\mathbf{u}_n}^T\mathbf{\Sigma}_{\mathbf{u}_n}^{-1}\mathbf{u}_n,\\[-15pt]$$
$$\mathbf{\Sigma}_n=\mathbf{\Sigma}_{\mathbf{f}_n}-\mathbf{\Sigma}_{\mathbf{f}_n\mathbf{u}_n}^T\mathbf{\Sigma}_{\mathbf{u}_n}^{-1}\mathbf{\Sigma}_{\mathbf{f}_n\mathbf{u}_n}.\\[-15pt]$$
Hence, $\tilde{\mathbf{\Sigma}}_n$ and $\tilde{\boldsymbol{\mu}}_n$ can be simplified as \\[-25pt]
$$\tilde{\boldsymbol{\mu}}_n=\mathbf{\Sigma}_{\mathbf{f}_n\mathbf{u}_n}^T\mathbf{\Sigma}_{\mathbf{u}_n}^{-1}\mathbf{m}_n,\\[-15pt]$$
$$\tilde{\mathbf{\Sigma}}_n=\mathbf{\Sigma}_n + \mathbf{\Sigma}_{\mathbf{f}_n\mathbf{u}_n}^T\mathbf{\Sigma}_{\mathbf{u}_n}^{-1} \mathbf{s}_n \mathbf{\Sigma}_{\mathbf{u}_n}^{-1}\mathbf{\Sigma}_{\mathbf{f}_n\mathbf{u}_n}.\\[-15pt]$$

It is clear that the conditional covariance matrix $\mathbf{\Sigma}_n$ does not involve $\mathbf{u}_n$, whereas $\boldsymbol{\mu}_n$ is a linear function of $\mathbf{u}_n$. Since $q(\mathbf{u}_n)=\Nagr(\mathbf{m}_n,\mathbf{s}_n)$, so we can write $q(\mathbf{f}_n)$ as \\[-45pt]
\begin{align*}
q(\mathbf{f}_n)&=\int \PP(\mathbf{f}_n\vert\mathbf{u}_n;\mathbf{f}_{n-1},\mathbf{X},\mathbf{Z}_{n},\boldsymbol{\delta}_{\mathbf{Z}_n},\alpha) q(\mathbf{u}_n) d\mathbf{u_n}\\
&=\int \frac{1}{{2\pi}^{(n_{S}+m)/2}|\mathbf{\Sigma}_n|^{1/2}|\mathbf{s}_n|^{1/2}} ~exp(-\frac{1}{2}Q)d\mathbf{u_n}\\
&=\frac{1}{{2\pi}^{(n_{S}+m)/2}|\mathbf{\Sigma}_n|^{1/2}|\mathbf{s}_n|^{1/2}}\int exp(-\frac{1}{2}Q)d\mathbf{u_n},
\end{align*}
 where \\[-40pt]
 \begin{align*}
 Q&=[(\mathbf{f}_n-\boldsymbol{\mu}_n)^T\mathbf{\Sigma}_n^{-1}(\mathbf{f}_n-\boldsymbol{\mu}_n)]+[(\mathbf{u}_n-\mathbf{m}_n)^T\mathbf{s}_n^{-1}(\mathbf{u}_n-\mathbf{m}_n)]\\
 &=[(\mathbf{f}_n-\mathbf{\Sigma}_{\mathbf{f}_n\mathbf{u}_n}^T\mathbf{\Sigma}_{\mathbf{u}_n}^{-1}\mathbf{u}_n)^T\mathbf{\Sigma}_n^{-1}(\mathbf{f}_n-\mathbf{\Sigma}_{\mathbf{f}_n\mathbf{u}_n}^T\mathbf{\Sigma}_{\mathbf{u}_n}^{-1}\mathbf{u}_n)]+[(\mathbf{u}_n-\mathbf{m}_n)^T\mathbf{s}_n^{-1}(\mathbf{u}_n-\mathbf{m}_n)]\\
 &=A+B, \\[-50pt]
\end{align*}
and \\[-30pt] 
 $$A=\mathbf{f}_n^T\mathbf{\Sigma}_n^{-1}\mathbf{f}_n-2 \mathbf{f}_n^T\mathbf{\Sigma}_n^{-1}\mathbf{\Sigma}_{\mathbf{f}_n\mathbf{u}_n}^T\mathbf{\Sigma}_{\mathbf{u}_n}^{-1}\mathbf{u}_n + \mathbf{u}_n^T\mathbf{\Sigma}_{\mathbf{u}_n}^{-1}\mathbf{\Sigma}_{\mathbf{f}_n\mathbf{u}_n}\mathbf{\Sigma}_n^{-1}\mathbf{\Sigma}_{\mathbf{f}_n\mathbf{u}_n}^T\mathbf{\Sigma}_{\mathbf{u}_n}^{-1}\mathbf{u}_n,\\[-15pt] $$
 $$B=\mathbf{u}_n^T\mathbf{s}_n^{-1}\mathbf{u}_n-2\mathbf{u}_n^T\mathbf{s}_n^{-1}\mathbf{m}_n+\mathbf{m}_n^T\mathbf{s}_n^{-1}\mathbf{m}_n.\\[-15pt] $$
Since the first term of $A$ and the last term of $B$ do not depend on ${\mathbf{u}_n}$, so $q(\mathbf{f}_n)$ can reach to this form \\[-25pt] 
$$q(\mathbf{f}_n)=\frac{1}{{2\pi}^{(n_{S}+m)/2}|\mathbf{\Sigma}_n|^{1/2}|\mathbf{s}_n|^{1/2}}exp(-\frac{1}{2}[\mathbf{f}_n^T\mathbf{\Sigma}_n^{-1}\mathbf{f}_n+\mathbf{m}_n^T\mathbf{s}_n^{-1}\mathbf{m}_n])\int exp(-\frac{1}{2}Q^\prime)d\mathbf{u},\\[-15pt] 
$$
where \\[-45pt] 
\begin{align*}
Q^\prime&=-2 \mathbf{f}_n^T\mathbf{\Sigma}_n^{-1}\mathbf{\Sigma}_{\mathbf{f}_n\mathbf{u}_n}^T\mathbf{\Sigma}_{\mathbf{u}_n}^{-1}\mathbf{u}_n + \mathbf{u}_n^T\mathbf{\Sigma}_{\mathbf{u}_n}^{-1}\mathbf{\Sigma}_{\mathbf{f}_n\mathbf{u}_n}\mathbf{\Sigma}_n^{-1}\mathbf{\Sigma}_{\mathbf{f}_n\mathbf{u}_n}^T\mathbf{\Sigma}_{\mathbf{u}_n}^{-1}\mathbf{u}_n+\mathbf{u}_n^T\mathbf{s}_n^{-1}\mathbf{u}_n-2\mathbf{u}_n^T\mathbf{s}_n^{-1}\mathbf{m}_n\\
&=-2[\mathbf{f}_n^T\mathbf{\Sigma}_n^{-1}\mathbf{\Sigma}_{\mathbf{f}_n\mathbf{u}_n}^T\mathbf{\Sigma}_{\mathbf{u}_n}^{-1}+\mathbf{m}_n^T\mathbf{s}_n^{-1}]\mathbf{u}_n + \mathbf{u}_n^T[\mathbf{\Sigma}_{\mathbf{u}_n}^{-1}\mathbf{\Sigma}_{\mathbf{f}_n\mathbf{u}_n}\mathbf{\Sigma}_n^{-1}\mathbf{\Sigma}_{\mathbf{f}_n\mathbf{u}_n}^T\mathbf{\Sigma}_{\mathbf{u}_n}^{-1}+\mathbf{s}_n^{-1}] \mathbf{u}_n.\\[-45pt] 
\end{align*}
It follows that \\[-25pt] 
$$\frac{1}{2}Q^\prime=-[\mathbf{f}_n^T\mathbf{\Sigma}_n^{-1}\mathbf{\Sigma}_{\mathbf{f}_n\mathbf{u}_n}^T\mathbf{\Sigma}_{\mathbf{u}_n}^{-1}+\mathbf{m}_n^T\mathbf{s}_n^{-1}]\mathbf{u}_n + \frac{1}{2} \mathbf{u}_n^T[\mathbf{\Sigma}_{\mathbf{u}_n}^{-1}\mathbf{\Sigma}_{\mathbf{f}_n\mathbf{u}_n}\mathbf{\Sigma}_n^{-1}\mathbf{\Sigma}_{\mathbf{f}_n\mathbf{u}_n}^T\mathbf{\Sigma}_{\mathbf{u}_n}^{-1}+\mathbf{s}_n^{-1}] \mathbf{u}_n.\\[-15pt] $$

The multivariate generalization of a mathematical trick known as "completion of squares" says that for a symmetric, non-singular matrix $\mathbf{A}$, the quadratic function can be written as \\[-30pt] 
$$\frac{1}{2}\mathbf{Z}^T\mathbf{A}\mathbf{Z} -\mathbf{b}^T\mathbf{Z}+\mathbf{C}=\frac{1}{2}(\mathbf{Z}-\mathbf{A}^{-1}\mathbf{b})^T\mathbf{A}(\mathbf{Z}-\mathbf{A}^{-1}\mathbf{b})-\frac{1}{2}\mathbf{b}^T\mathbf{A}^{-1}\mathbf{b}+\mathbf{C}.\\[-15pt] $$
Now, we can apply this trick in our situation by these assumptions \\[-25pt] 
$$\mathbf{Z}:=\mathbf{u}_n~~,~~\mathbf{C}:=\mathbf{0},\\[-15pt] $$
$$\mathbf{b}:=[\mathbf{f}_n^T\mathbf{\Sigma}_n^{-1}\mathbf{\Sigma}_{\mathbf{f}_n\mathbf{u}_n}^T\mathbf{\Sigma}_{\mathbf{u}_n}^{-1}+\mathbf{m}_n^T\mathbf{s}_n^{-1}]^T,\\[-15pt] $$
$$\mathbf{A}:=\mathbf{\Sigma}_{\mathbf{u}_n}^{-1}\mathbf{\Sigma}_{\mathbf{f}_n\mathbf{u}_n}\mathbf{\Sigma}_n^{-1}\mathbf{\Sigma}_{\mathbf{f}_n\mathbf{u}_n}^T\mathbf{\Sigma}_{\mathbf{u}_n}^{-1}+\mathbf{s}_n^{-1}.\\[-10pt] $$
By application of Aitken's integral and this fact that $\mathbf{b}$ in our setting does not involve $\mathbf{u}_n$, $q(\mathbf{f}_n)$ can be simplified as \\[-45pt] 
\begin{align*}
q(\mathbf{f}_n)&=\frac{{2\pi}^{m/2}|\mathbf{A}^{-1}|^{1/2}}{{2\pi}^{(n_{S}+m)/2}|\mathbf{\Sigma}_n|^{1/2}|\mathbf{s}_n|^{1/2}}exp(-\frac{1}{2}[\mathbf{f}_n^T\mathbf{\Sigma}_n^{-1}\mathbf{f}_n+\mathbf{m}_n^T\mathbf{s}_n^{-1}\mathbf{m}_n-\mathbf{b}^T\mathbf{A}^{-1}\mathbf{b}])\\
&=\frac{|\mathbf{A}^{-1}|^{1/2}}{{2\pi}^{n_{S}/2}|\mathbf{\Sigma}_n|^{1/2}|\mathbf{s}_n|^{1/2}}exp(-\frac{1}{2}[\mathbf{f}_n^T\mathbf{\Sigma}_n^{-1}\mathbf{f}_n+\mathbf{m}_n^T\mathbf{s}_n^{-1}\mathbf{m}_n-\mathbf{b}^T\mathbf{A}^{-1}\mathbf{b}]).\\[-45pt] 
\end{align*}
Now by showing \\[-25pt] 
$$[(\mathbf{f}_n-\tilde{\boldsymbol{\mu}}_n)^T\tilde{\mathbf{\Sigma}}_n^{-1}(\mathbf{f}_n-\tilde{\boldsymbol{\mu}}_n)]=[\mathbf{f}_n^T\mathbf{\Sigma}_n^{-1}\mathbf{f}_n+\mathbf{m}_n^T\mathbf{s}_n^{-1}\mathbf{m}_n-\mathbf{b}^T\mathbf{A}^{-1}\mathbf{b}],\\[-15pt] $$
and 
$$|\tilde{\mathbf{\Sigma}}_n|^{-1/2}=\frac{|\mathbf{A}^{-1}|^{1/2}}{|\mathbf{\Sigma}_n|^{1/2}|\mathbf{s}_n|^{1/2}},$$
the proof will be completed. To proof the first equation, we start from the left hand side (LHS) to reach the the right hand side (RHS). Using the Woodbury identity, $\tilde{\mathbf{\Sigma}}_n^{-1}$ can be writen as \\[-45pt] 
\begin{align*}
\tilde{\mathbf{\Sigma}}_n^{-1}&=\mathbf{\Sigma}_n^{-1}-\mathbf{\Sigma}_n^{-1}\mathbf{\Sigma}_{\mathbf{f}_n\mathbf{u}_n}^T\mathbf{\Sigma}_{\mathbf{u}_n}^{-1}\big(\mathbf{\Sigma}_{\mathbf{u}_n}^{-1}\mathbf{\Sigma}_{\mathbf{f}_n\mathbf{u}_n}\mathbf{\Sigma}_n^{-1}\mathbf{\Sigma}_{\mathbf{f}_n\mathbf{u}_n}^T\mathbf{\Sigma}_{\mathbf{u}_n}^{-1}+\mathbf{s}_n^{-1}\big)^{-1}\mathbf{\Sigma}_{\mathbf{u}_n}^{-1}\mathbf{\Sigma}_{\mathbf{f}_n\mathbf{u}_n}\mathbf{\Sigma}_n^{-1}\\
&=\mathbf{\Sigma}_n^{-1}-\mathbf{\Sigma}_n^{-1}\mathbf{\Sigma}_{\mathbf{f}_n\mathbf{u}_n}^T\mathbf{\Sigma}_{\mathbf{u}_n}^{-1}\mathbf{A}^{-1}\mathbf{\Sigma}_{\mathbf{u}_n}^{-1}\mathbf{\Sigma}_{\mathbf{f}_n\mathbf{u}_n}\mathbf{\Sigma}_n^{-1}.\\[-45pt] 
\end{align*}
By plugging $\tilde{\boldsymbol{\mu}}_n$ and $\tilde{\mathbf{\Sigma}}_n^{-1}$ into the LHS and using equations \\[-25pt] 
$$\mathbf{\Sigma}_{\mathbf{u}_n}^{-1}\mathbf{\Sigma}_{\mathbf{f}_n\mathbf{u}_n}\mathbf{\Sigma}_n^{-1}\mathbf{f}_n=\mathbf{b}-\mathbf{s}_n^{-1}\mathbf{m}_n,\\[-15pt] $$
and \\[-25pt] 
$$\mathbf{\Sigma}_{\mathbf{u}_n}^{-1}\mathbf{\Sigma}_{\mathbf{f}_n\mathbf{u}_n}\mathbf{\Sigma}_n^{-1}\mathbf{\Sigma}_{\mathbf{f}_n\mathbf{u}_n}^T\mathbf{\Sigma}_{\mathbf{u}_n}^{-1}=\mathbf{A}-\mathbf{s}_n^{-1},\\[-15pt] $$
we reach\\[-45pt] 
\begin{align*}
LHS&=\mathbf{f}_n^T\mathbf{\Sigma}_n^{-1}\mathbf{f}_n-(\mathbf{b}-\mathbf{s}_n^{-1}\mathbf{m}_n)^T\mathbf{m}_n-(\mathbf{b}-\mathbf{s}_n^{-1}\mathbf{m}_n)^T\mathbf{A}^{-1}(\mathbf{b}-\mathbf{s}_n^{-1}\mathbf{m}_n)\\
&+2\mathbf{m}_n^T(\mathbf{A}-\mathbf{s}_n^{-1})\mathbf{A}^{-1}(\mathbf{b}-\mathbf{s}_n^{-1}\mathbf{m}_n)-\mathbf{m}_n^T(\mathbf{b}-\mathbf{s}_n^{-1}\mathbf{m}_n)+\mathbf{m}_n^T(\mathbf{A}-\mathbf{s}_n^{-1})\mathbf{m}_n\\
&-\mathbf{m}_n^T(\mathbf{A}-\mathbf{s}_n^{-1})\mathbf{A}^{-1}(\mathbf{A}-\mathbf{s}_n^{-1})\mathbf{m}_n\\
&=\mathbf{f}_n^T\mathbf{\Sigma}_n^{-1}\mathbf{f}_n+\mathbf{m}_n^T\mathbf{s}_n^{-1}\mathbf{m}_n-\mathbf{b}^T\mathbf{A}^{-1}\mathbf{b}\\
&=RHS.
\end{align*}
The next equation can be written as \\[-45pt] 
\begin{align*}
|\tilde{\mathbf{\Sigma}}_n^{-1}|^{1/2}
&=|\mathbf{A}^{-1}|^{1/2}|\mathbf{\Sigma}_n^{-1}|^{1/2}|\mathbf{s}_n^{-1}|^{1/2},\\[-55pt] 
\end{align*}
using this fact that\\[-25pt] 
$$\tilde{\mathbf{\Sigma}}_n^{-1}=\mathbf{\Sigma}_n^{-1}-\mathbf{\Sigma}_n^{-1}\mathbf{\Sigma}_{\mathbf{f}_n\mathbf{u}_n}^T\mathbf{\Sigma}_{\mathbf{u}_n}^{-1}\mathbf{A}^{-1}\mathbf{\Sigma}_{\mathbf{u}_n}^{-1}\mathbf{\Sigma}_{\mathbf{f}_n\mathbf{u}_n}\mathbf{\Sigma}_n^{-1}.\\[-15pt] $$
\end{proof}

Derivation of the ELBO of our DGP in \ref{eq:dgpelbo}
\begin{proof} The evidence lower bound of our DGP, $\Lagr_{DGP}$, is defined as following\\[-20pt] 
$$\Lagr_{DGP}=\EE_{q(\{\mathbf{f}_n, \mathbf{u}_n\}_{n=1}^{N},\alpha)}\textrm{log}\Bigg(\frac{\PP(\mathbf{y},\{\mathbf{f}_n, \mathbf{u}_n\}_{n=1}^{N},\alpha)}{q(\{\mathbf{f}_n, \mathbf{u}_n\}_{n=1}^{N},\alpha)}\Bigg)$$
By substituting (\ref{eq:jp3}), (\ref{eq:12al}) in the RHS of the above equation, $\Lagr_{DGP}$ can be written as \\[-20pt] 
$$=\int \cdots \int q(\{\mathbf{f}_n,
\mathbf{u}_n\}_{n=1}^{N},\alpha)~ \textrm{log}\Bigg(\PP(\mathbf{y}\vert \mathbf{f}_N) \times\frac{\PP(\mathbf{u}_1;\mathbf{Z}_{1})}{q(\mathbf{u}_1)} \times \frac{\prod_{n=2}^{N}\PP(\mathbf{u}_n\vert\mathbf{Z}_{n},\boldsymbol{\delta}_{\mathbf{Z}_n},\alpha)}{\prod_{n=2}^{N}q(\mathbf{u}_n)} \times \frac{\PP(\alpha)}{q(\alpha)}\Bigg)$$
$~~d\{\mathbf{f}_n, \mathbf{u}_n\}_{n=1}^{N} d\alpha $ \\[-20pt] 
$$=\int \cdots \int q(\{\mathbf{f}_n\}_{n=1}^{N},\alpha) ~ \textrm{log}\Big(\prod_{i=1}^{n_S}\PP(y_i\vert \mathbf{f}_{N}^i)\Big) ~d\{\mathbf{f}_n\}_{n=1}^{N} d\alpha + \int q(\mathbf{u}_1) \textrm{log} \Big(\frac{\PP(\mathbf{u}_1;\mathbf{Z}_{1})}{q(\mathbf{u}_1)}\Big) d\mathbf{u}_1$$
$$+\int \cdots\int q(\{\mathbf{u}_n\}_{n=2}^{N},\alpha) \textrm{log}\Bigg(\frac{\prod_{n=2}^{N}\PP(\mathbf{u}_n\vert\mathbf{Z}_{n},\boldsymbol{\delta}_{\mathbf{Z}_n},\alpha)}{\prod_{n=2}^{N}q(\mathbf{u}_n)}\Bigg)d\{ \mathbf{u}_n\}_{n=2}^{N}d\alpha + \int q(\alpha)\textrm{log}\Big(\frac{\PP(\alpha)}{q(\alpha)}\Big)d\alpha.$$
Therefore, the ELBO of the DGP can be formed as \\[-25pt] 

$$\Lagr_{DGP}=\sum_{i=1}^{n_S} \EE_{q(\{\mathbf{f}_n\}_{n=1}^{N},\alpha)}\Big(~\textrm{log}~\PP(y_i\vert \mathbf{f}_{N}^i)~\Big) - \textrm{KL}\Big(q(\mathbf{u}_1)\parallel \PP(\mathbf{u}_1;\mathbf{Z}_{1})\Big)\\[-10pt] $$

$$~~~~~~~- \EE_{q(\alpha)}\Big[ \sum_{n=2}^{N}\textrm{KL}\Big(q(\mathbf{u}_n) \parallel \PP(\mathbf{u}_n;\mathbf{Z}_{n},\boldsymbol{\delta}_{\mathbf{Z}_n},\alpha)\Big)\Big]-\textrm{KL}\Big(q(\alpha) \parallel \PP(\alpha)\Big).$$
\end{proof}

\end{document}